\documentclass{IEEEtran}
\usepackage{subcaption}

\input{def_commands.tex}

\title{Compositional Small-Gain and Small-Phase Stability Analysis}
\author{Anton Ponomarev\thanks{Anton Ponomarev is with the Karlsruhe Institute of Technology, Germany (e-mail: anton.ponomarev@kit.edu).}}

\begin{document}

\maketitle

\begin{abstract}
   We adapt the small-gain and small-phase stability analysis to systems composed of MIMO LTI subsystems via an iteration of the series, parallel, and feedback interconnections. Based on a certain set of parameters (including the phase sector and the Crawford number) of the subsystems, we bound the same set of parameters of each consecutive interconnection. This is called ``composition of bounds\rlap{.}'' Composing the gains of subsystems in the same manner enables the small-gain-or-phase analysis of fairly complex systems. The method is illustrated with two examples that are not tractable by the classical small-phase analysis, and one other example that nevertheless benefits from the compositional approach.
\end{abstract}

\section{Introduction}

The \emph{small-gain and small-phase} stability analysis of MIMO LTI systems~\cite{zhaoWhenSmallGain2022} combines the classical \emph{small-gain} theorem~\cite{zamesInputoutputStabilityTimevarying1966a} and the \emph{small-phase} theorem~\cite{wangPhasesComplexMatrix2020}. It has already found an application in the analysis of large-scale electrical grids~\cite{huangGainPhaseDecentralized2024}.

\subsection{Does Small-Phase Theorem Apply to Complex Systems?}

Since its inception, the small-gain technique has been successfully extended to complex networks~\cite{vidyasagarInputoutputAnalysisLargescale1981,dashkovskiySmallGainTheorems2010}, but similar extension of the small-phase theorem meets with a difficulty; let us explain it at an intuitive level. In the linear case, the gain is a matrix norm $\norm{\cdot}$. The small-gain theorem relies on the inclusion of the spectrum $\sigma(\cdot)$ of the open-loop transfer matrix $\bfA$ in the norm-disk:
\begin{equation}
   \sigma(\bfA) \subseteq \set[\big]{z \in \bbC \colon \abs{z} \leq \norm{\bfA}}.
\end{equation}
The small-phase theorem similarly uses a spectral inclusion but this time in a sector of the complex plane covered by the \emph{angular field of values} $\calW'(\bfA)$~\cite[Definition~1.1.2]{hornTopicsMatrixAnalysis1991}:
\begin{equation}
   \label{eq: spectral inclusion sector}
   \sigma(\bfA) \subseteq \calW'(\bfA).
\end{equation}
If there are two matrices in the loop $(\bfA = \bfA_1 \bfA_2)$, the small-gain theorem employs the submultiplicativity of the matrix norm to bound the open-loop gain as
\begin{equation}
   \label{eq: norm submultiplicativity}
   \norm{\bfA_1 \bfA_2} \leq \norm{\bfA_1} \norm{\bfA_2}.
\end{equation}
The angular field of values, however, is not submultiplicative in general:
\begin{equation}
   \label{eq: angular field not submultiplicative}
   \calW'(\bfA_1 \bfA_2) \not\subseteq \calW'(\bfA_1) \calW'(\bfA_2).
\end{equation}
The small-phase theorem instead resorts to another spectral inclusion~\cite[Theorem~1.7.8]{hornTopicsMatrixAnalysis1991}
\begin{equation}
   \label{eq: spectral inclusion sector 2}
   \sigma(\bfA_1 \bfA_2) \subseteq \calW'(\bfA_1) \calW'(\bfA_2)
\end{equation}
(assuming that at least one matrix is sectorial). With three or more subsystems in the loop, the small-gain theorem easily generalizes by \emph{iteration} of~\eqref{eq: norm submultiplicativity}. The sectorial inclusion~\eqref{eq: spectral inclusion sector 2}, on the other hand, does not allow iteration. This is the obstacle that halts the sequence~\eqref{eq: spectral inclusion sector}, \eqref{eq: spectral inclusion sector 2}, \dots

\subsection{Segmental Phase}

Recently, a \emph{segmental} variant of the small-phase theorem has been introduced~\cite{chenCyclicSmallPhase2026}. Roughly speaking, it generalizes the spectral inclusion~\eqref{eq: spectral inclusion sector 2} to more subsystems but with a ``product of segments'' on the right-hand side~-- for the actual statement, see~\cite[Theorem~1]{chenCyclicSmallPhase2026}. It should be noted that, whereas the classical small-phase theorem is only meaningful for sectorial matrices, its segmental version does not have this limitation.

\subsection{Proposed Approach}

In the present study, we take a different route: instead of looking for a way to continue the sequence~\eqref{eq: spectral inclusion sector}, \eqref{eq: spectral inclusion sector 2}, \dots, we try to establish an analogue of the submultiplicativity property~\eqref{eq: norm submultiplicativity} for the classical sectorial phase definition.

Why should such an analogue be attainable? The intuitive grounds are seen from two trivial observations:
\begin{enumerate}
   \item The inclusion in~\eqref{eq: angular field not submultiplicative} does hold if the matrices are scalar, so it should hold (approximately) if they are ``close to scalar'' as well. In the phase theory context, a natural way to formalize a matrix being ``close to scalar'' is to say that its \emph{numerical range} is ``not too spread out\rlap{.}''
   \item What the examples of non-submultiplicativity~\eqref{eq: angular field not submultiplicative} in~\cite[Sec.~1.7]{hornTopicsMatrixAnalysis1991} have in common is the fact that the numerical range there contains the origin. This suggests that submultiplicativity may hold (approximately) if the numerical range is ``far enough from the origin\rlap{.}''
\end{enumerate}
We capture the notions of the numerical range being ``not too spread out'' and ``far enough from the origin'' in the concept of the \emph{Crawford defect} (Definition~\ref{de: crawford defect}) built on top of the \emph{Crawford number} $c(\cdot)$ (the distance from the origin to the numerical range). With this notion, an analogue of~\eqref{eq: norm submultiplicativity} is established as follows: we define a set of 6 functions $\calS(\cdot)$ of a matrix and show that $\calS(\bfA_1 \bfA_2)$ can be bounded in terms of $\calS(\bfA_1)$ and $\calS(\bfA_2)$. In particular, $\calS(\cdot)$ includes the Crawford number $c(\cdot)$ and the classical phase sector $\Phi(\cdot)$.  This result is stated as Theorem~\ref{th: series}~-- our main technical contribution.

Under the restrictions of Theorem~\ref{th: series} (essentially, the product of the Crawford defects should be small), the standard small-phase theorem can now be applied to an arbitrary number of subsystems in the feedback loop.

The approach is extended beyond a single loop via the \emph{compositional} idea. Note that the gain of a series, parallel, or feedback interconnection of several subsystems can be bounded via the gains of the subsystems~-- we call it the \emph{composition} of the norm bounds\footnote{We cannot provide a reference to an explicit formulation of the compositional small-gain analysis, but there are cases of it being used, e.g., \cite{mauryaMarineVehiclePath2009}.} (Theorem~\ref{th: norm compositions}) and define a similar operation for the bounds on our set $\calS(\cdot)$. In addition to the already mentioned series composition (Theorem~\ref{th: series}), we establish the parallel (Theorem~\ref{th: parallel}) and feedback composition results (Section~\ref{se: feedback}). With these, the small-gain and small-phase stability conditions can be applied to an iterated composition of series, parallel, and feedback interconnections~-- this is our main conceptual contribution.

\subsection{Some Comments}

It appears that our results are closely related to the segmental phase idea mentioned above. These relations are discussed in Section~\ref{se: normalized range} where we suggest a potential way to combine the sectorial and segmental approaches. This may reduce the conservatism as illustrated in the example of Section~\ref{se: example 3}.

Another example (Section~\ref{se: example 1}) demonstrates an application of the compositional approach to a single-loop system with near-unity feedback where the small-gain and small-phase theorems would normally be applied directly. However, with the compositional approach we find a larger stability domain: it turns out that the compositional \emph{small}-gain and \emph{small}-phase analysis in this case asserts stability under a \emph{high}-gain condition\footnote{Although a high-gain system may raise questions of noise sensitivity, there exist established methods of high-gain stabilization~\cite{ilchmannNonidentifierbasedHighgainAdaptive1993,bergerFunnelControlNonlinear2021}.}. This fact may seem surprising but has an intuitive explanation: we deal with the notion of the numerical range being ``far enough from the origin'' which essentially means that ``the gain is high'' in a strong sense: it implies that the norm is large but is a stronger property than that. By way of the inversion of the forward-path matrix in the feedback formula, high gain turns into a small one.

The latter example reinforces the crucial role played by the Crawford number $c(\cdot)$: it quantifies how ``strongly high-gain'' a matrix is. Together with the numerical radius $r(\cdot)$ and phase sector $\Phi(\cdot)$, the Crawford number bounds the numerical range to a ring sector~-- these bounds we call \emph{ring-sectorial} (Fig.~\ref{fig: ring-sectorial bounds}). Replacing the pure sectorial bounds of the small-phase theory with \emph{ring}-sectorial bounds underpins our approach.

\subsection{Structure of the Paper}

Section~\ref{se: notation} introduces the notation; the unusual part here may be the definitions of phases, phase arcs, and operations on them. In Section~\ref{se: problem}, the considered set of systems (\emph{decomposable} systems) is defined, and the stability problem (\emph{internal} stability) is stated. In Section~\ref{se: preliminaries} we recall the familiar concepts of the small-gain and small-phase analysis. Section~\ref{se: main} presents the main results: compositional small-gain-or-phase analysis and the rules of composition. The results are discussed in Section~\ref{se: discussion}; we mainly mention the conservatism and possible ways to reduce it. Section~\ref{se: examples} contains three examples that are simple enough for analytical results to be attainable.

\section{Notation}
\label{se: notation}

\subsection{Standard Notation}

\begin{itemize}
   \item $\rmI$~-- $n \times n$ identity matrix;
   \item $j$~-- complex unit;
   \item $\norm{\cdot}$~-- 2-norm;
   \item $(\cdot)^*$~-- conjugate transpose;
   \item $\bfv \perp \bfw \iff \bfv^*\bfw = 0$;
   \item $\RHinfnn$~-- the space of real rational, proper, and stable $n\times n$ transfer matrices.
\end{itemize}
Given $S, S_1, S_2 \subset \bbC$:
\begin{itemize}
   \item $S^* = \set{z^* \colon z \in S}$;
   \item $\conv(S)$~-- convex hull;
   \item $\dist(S_1, S_2) = \inf\set[\big]{\abs{z_1-z_2} \colon z_{1,2} \in S_{1,2}}$;
   \item $S_1 + S_2 = \set{z_1 + z_2 \colon z_{1,2} \in S_{1,2}}$;
   \item $S_1 S_2 = \set{z_1 z_2 \colon z_{1,2} \in S_{1,2}}$.
\end{itemize}
The symbols $\conv$, $\dist$, and $+$ are redefined when they apply to phases and the arguments of complex numbers: there we adopt the following formalism.

\subsection{Phase-Related Notation}
\label{se: phases}

Typically, the object called ``a phase'' emerges as the argument of some complex number. Practical computations with phases require careful handling of the argument branches. The following notation is intended to circumvent these details.

Let the quotient group $\bbT = \bbR/(2\pi\bbZ)$ represent the space of \emph{phases}. As an equivalence class, phase $\theta \in \bbT$ has a representative $x \in \bbR$, denoted $\theta = [x]$.

The \emph{phase arc} covered by an interval $[x_1, x_2] \subset \bbR$ is
\begin{equation}
   \arc[x_1, x_2]
   = \set[\big]{[x] \colon x \in [x_1, x_2]} \subseteq \bbT.
\end{equation}

The following definitions are motivated by the isometry $[x] \mapsto \ee^{jx}$ between $\bbT$ and the unit circle $S^1 \subset \bbC$ endowed with the geodesic distance (the length of the shortest arc).

\emph{Metric} $\dist(\cdot, \cdot)$ is defined on $\bbT$ as
\begin{equation}
   \dist(\theta_1, \theta_2) = \min\set[\big]{\abs{x_1 - x_2} \colon
   [x_1] = \theta_1, [x_2] = \theta_2}.
\end{equation}
The distance between two sets $\Theta_1, \Theta_2 \subseteq \bbT$ is
\begin{equation}
   \dist(\Theta_1, \Theta_2) = \inf\set[\big]{\dist(\theta_1, \theta_2) \colon
   \theta_1 \in \Theta_1, \theta_2 \in \Theta_2}.
\end{equation}
The \emph{diameter} of a set $\Theta \subseteq \bbT$ is
\begin{equation}
   \diam\Theta = \sup\set[\big]{\dist(\theta_1, \theta_2) \colon
   \theta_1, \theta_2 \in \Theta}.
\end{equation}
The Hausdorff measure induced on $\bbT$ by $\dist(\cdot, \cdot)$ is called \emph{length} and denoted $\len(\cdot)$. In particular,
\begin{align}
   \diam\arc[x_1, x_2] &= \min\set{\pi, x_2 - x_1}, \\
   \len\arc[x_1, x_2] &= \min\set{2\pi, x_2 - x_1}.
\end{align}

Set $\Theta \subseteq \bbT$ is called \emph{convex} if every pair of its points can be connected by a minimizing $\bbT$-geodesic contained in $\Theta$. Among the closed subsets of $\bbT$, the only convex sets are $\bbT$ itself and $\arc[x_1, x_2]$ with $x_1 \leq x_2 \leq x_1 + \pi$. The \emph{convex hull} $\conv\Theta$ is the minimal convex superset of $\Theta$.

The \emph{sum} of two phase sets $\Theta_1, \Theta_2 \subseteq \bbT$ is
\begin{equation}
   \Theta_1 + \Theta_2 = \set[\big]{[x_1 + x_2] \colon
   [x_1] \in \Theta_1, [x_2] \in \Theta_2}
\end{equation}
and their \emph{join} is
\begin{equation}
   \Theta_1 \vee \Theta_2 = \conv(\Theta_1 \cup \Theta_2).
\end{equation}

\begin{figure}
   \centering
   \begin{subfigure}[b]{80mm}
      \centering
      \includegraphics{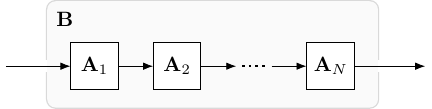}
      \caption{Series: $\bfB = \bfA_N \bfA_{N-1} \dots \bfA_1$.}
      \label{fig: connection series}
   \end{subfigure}\\
   \vspace{4mm}
   \begin{subfigure}[b]{80mm}
      \centering
      \includegraphics{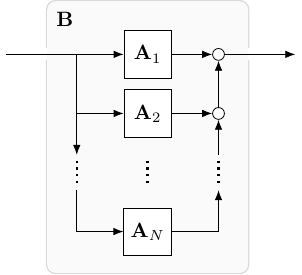}
      \caption{Parallel: $\bfB = \bfA_1 + \bfA_2 + \dots + \bfA_N$.}
      \label{fig: connection parallel}
   \end{subfigure}\\
   \vspace{4mm}
   \begin{subfigure}[b]{80mm}
      \centering
      \includegraphics{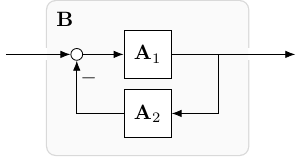}
      \caption{Feedback: $\bfB = (\rmI + \bfA_1 \bfA_2)\inv \bfA_1$.}
      \label{fig: connection feedback}
   \end{subfigure}
   \caption{Elementary interconnections.}
   \label{fig: connections}
\end{figure}

\begin{figure}
   \centering
   \begin{subfigure}[b]{82mm}
      \centering
      \includegraphics{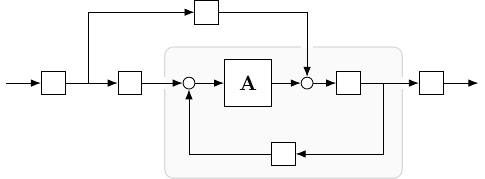}
      \caption{}
   \end{subfigure}\\
   \vspace{4mm}
   \begin{subfigure}[b]{82mm}
      \centering
      \includegraphics{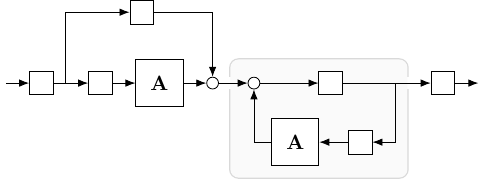}
      \caption{}
   \end{subfigure}
   \caption{System (a) is indecomposable because the feedback loop has two entry points (indeed, an $n \times 2n$ system cannot be obtained from $n \times n$ systems via elementary interconnections). However, the system can be transformed to the decomposable form (b) by pulling block $\bfA$ back before the adder (\emph{node splitting}) and swapping the two adders. See also Remark~\ref{re: node splitting}.}
   \label{fig: node splitting}
\end{figure}

\section{Problem Statement}
\label{se: problem}

In this paper we consider MIMO LTI systems and their interconnections (wirings) represented by block diagrams. For convenience, the terms ``system $\bfA$'' and ``transfer matrix $\bfA(s)$'' are used interchangeably.

\begin{definition}[Elementary interconnection]
   Consider $n \times n$ LTI systems $\bfA_1, \bfA_2, \dots, \bfA_N$. Their series, parallel, and feedback interconnections shown in Fig.~\ref{fig: connections} are called elementary.
\end{definition}

Systems resulting from elementary interconnections are again $n \times n$ and can participate in further interconnections.

\begin{definition}[Decomposable system]
   The system obtained as a composition of elementary interconnections is called decomposable.
\end{definition}

Some indecomposable systems can be rendered decomposable via block diagram transformations (for example, the \emph{node splitting}\footnote{The name is inspired by~\cite{janssenMakingGraphsReducible1997}.} transformation in Fig.~\ref{fig: node splitting}).

\begin{definition}[Internal stability]
   Given a control system in the block diagram form, consider its state-space realization obtained by wiring the minimal realizations of the blocks. If this realization is well-posed and asymptotically stable under zero input then the system is called internally stable (by extension of \cite[Definition~5.2]{zhouRobustOptimalControl1996}).
\end{definition}

\begin{remark}
   \label{re: node splitting}
   Some block diagram transformations (e.g., the one in Fig.~\ref{fig: node splitting}) rely on the duplication of subsystems and, therefore, of the state variables. Since the original system is a state-space restriction of the transformed one, internal stability of the latter implies that the original system is internally stable as well.
\end{remark}

We are concerned with the following.

\begin{problem}
   \label{problem}
   Given a decomposable LTI system constructed from $\RHinfnn$-subsystems, find sufficient conditions of its internal stability.
\end{problem}

\section{Preliminaries}
\label{se: preliminaries}

\subsection{Numerical Range and its Ring-Sectorial Bounds}
\label{se: ring sectorial bounds}

Let $\bfA \in \bbC^{n\times n}$. The \emph{numerical range} of $\bfA$ is the set
\begin{equation}
   \calW(\bfA) = \set*{
      \frac{\bfx^* \bfA \bfx}{\norm{\bfx}^2} \colon
      \bfx \in \bbC^n \setminus \set{0}
   } \subset \bbC
\end{equation}
which is closed, convex, and contains the spectrum of $\bfA$~\cite[Ch.~1]{hornTopicsMatrixAnalysis1991}. The \emph{numerical radius} of $\bfA$ is
\begin{equation}
   r(\bfA) = \max\set[\big]{
      \abs{z} \colon z \in \calW(\bfA)
   } \in [0, \norm{\bfA}],
\end{equation}
the \emph{Crawford number}~\cite{stewartPertubationBoundsDefinite1979} is
\begin{equation}
   c(\bfA) = \min\set[\big]{
      \abs{z} \colon z \in \calW(\bfA)
   } \in [0, r(\bfA)],
\end{equation}
and the \emph{phase sector}~\cite{zhaoWhenSmallGain2022} is the phase arc
\begin{equation}
   \Phi(\bfA) = \set[\big]{
      \arg{z} \colon z \in \calW(\bfA)
   } \subseteq \bbT.
\end{equation}
If $0 \in \calW(\bfA)$, we define $\Phi(\bfA) = \bbT$.

\begin{definition}[Sectorial matrix]
   Matrix $\bfA$ is called sectorial (with vertex $0$) if $0 \not\in \calW(\bfA)$~\cite[Section~V]{katoPerturbationTheoryLinear1984}. Two other equivalent definitions are $c(\bfA) > 0$ and $\diam\Phi(\bfA) < \pi$.
\end{definition}

A sectorial matrix is always invertible, and its inverse is also sectorial. The following relations will be useful.

\begin{lemma}
   \label{le: inverse bounds}
   Let $\bfA \in \bbC^{n\times n}$ be sectorial. Then
   \begin{subequations}
      \begin{gather}
         \frac{c(\bfA)}{\norm{\bfA}^2} \leq c(\bfA\inv)
         \leq \frac{r(\bfA)}{\norm{\bfA}^2}, \\
         \Phi(\bfA\inv) = -\Phi(\bfA).
      \end{gather}
   \end{subequations}
\end{lemma}

\begin{IEEEproof}
   See~\cite[Sec.~4]{hochstenbachNumericalApproximationField2013}.
\end{IEEEproof}

\begin{definition}[Ring-sectorial bounds]
   The numerical range $\calW(\bfA)$ is bounded by $r(\bfA)$, $c(\bfA)$, and $\Phi(\bfA)$ as shown in Fig.~\ref{fig: ring-sectorial bounds}. We call this type of bound ring-sectorial. The bounds in Fig.~\ref{fig: ring-sectorial bounds} are exact. Bounds $\hat r$, $\hat c$, $\hat\Phi$ are conservative if
   \begin{equation}
      \label{eq: conservative bounds}
      r(\bfA) \leq \hat r, \quad
      c(\bfA) \geq \hat c, \quad
      \Phi(\bfA) \subseteq \hat \Phi.
   \end{equation}
\end{definition}

\begin{remark}[Finding ring-sectorial bounds]
   Given a numerical matrix $\bfA$, its numerical range $\calW(\bfA)$ can be outer-bounded with arbitrary accuracy by a polygon~\cite{johnsonNumericalDeterminationField1978} yielding conservative ring-sectorial bounds. The phase sector can be computed directly as well~\cite[Section~2]{wangPhasesComplexMatrix2020}.
\end{remark}

\subsection{Stability of Elementary Interconnections}

If all subsystems $\bfA_k$ are of class $\RHinfnn$ then the series and parallel structures in Figs.~\ref{fig: connection series} and~\ref{fig: connection parallel} are internally stable; in particular, $\bfB \in \RHinfnn$.

For the feedback loop, we recall the following statement which is a special case of~\cite[Theorem~1]{zhaoWhenSmallGain2022}.

\begin{theorem}[Small-gain-or-phase]
   \label{th: small gain phase}
   Let $\bfA_1, \bfA_2 \in \RHinfnn$. The feedback loop in Fig.~\ref{fig: connection feedback} is internally stable, and in particular $\bfB \in \RHinfnn$, if at each frequency $\omega \in [0, \infty]$ at least one of the following conditions holds:
   \begin{subequations}
      \label{eq: small gain phase}
      \begin{align}
         \label{eq: small gain}
         & \norm{\bfA_1} \norm{\bfA_2} < 1 \\
         \label{eq: small phase}
         \text{or}\quad & \dist\p[\big]{[\pi],\: \Phi\p{\bfA_1} + \Phi\p{\bfA_2}} > 0
      \end{align}
   \end{subequations}
   where $\bfA_k = \bfA_k(j\omega)$.
\end{theorem}

Condition~\eqref{eq: small gain} is called \emph{small-gain}, and~\eqref{eq: small phase} is called \emph{small-phase} (``smaller than $\pi$'').

\subsection{Compositional Small-Gain Analysis}
\label{se: compositional small gain}

Norm bounds are \emph{composable}: under some conditions, the norm $\norm{\bfB(j\omega)}$ of every elementary interconnection in Fig.~\ref{fig: connections} can be bounded if bounds on $\norm{\bfA_k(j\omega)}$ are available. This is summarized as follows.

\begin{theorem}[Norm bound compositions]
   \label{th: norm compositions}
   Suppose that all transfer matrices $\bfA_k(s)$ in Fig.~\ref{fig: connections}, evaluated at a given point $s = j\omega$, have known norm bounds
   \begin{equation}
      \norm{\bfA_k} \leq \hat\gamma_k.
   \end{equation}
   Then for the series interconnection (Fig.~\ref{fig: connection series})
   \begin{subequations}
      \begin{align}
         \bfB = \bfA_N \bfA_{N-1} \dots \bfA_1 &\implies
         \norm{\bfB} \leq \prod_{k=1}^N \hat\gamma_k, \\
         \intertext{for the parallel interconnection (Fig.~\ref{fig: connection parallel})}
         \bfB = \bfA_1 + \bfA_2 + \dots + \bfA_N &\implies
         \norm{\bfB} \leq \sum_{k=1}^N \hat\gamma_k, \\
         \intertext{and for the feedback interconnection (Fig.~\ref{fig: connection feedback}), under the small-gain condition $\hat\gamma_1 \hat\gamma_2 < 1$,}
         \label{eq: norm composition feedback}
         \bfB = (\rmI + \bfA_1 \bfA_2)\inv \bfA_1 &\implies
         \norm{\bfB} \leq \frac{1}{\hat\gamma_1\inv - \hat\gamma_2}.
      \end{align}
   \end{subequations}
\end{theorem}

Combining the small-gain part~\eqref{eq: small gain} of Theorem~\ref{th: small gain phase} with Theorem~\ref{th: norm compositions}, the following \emph{compositional} small-gain approach to Problem~\ref{problem} can be conceived. The given system is represented as a hierarchy of elementary interconnections. Starting from the gains (frequency-wise norms) of the bottom-level subsystems, the gains of their elementary interconnections, and interconnections thereof, are upper-bounded by Theorem~\ref{th: norm compositions} while ascending the hierarchy. During this process, internal stability of each occurring feedback loop is verified by the small-gain condition~\eqref{eq: small gain}. If every loop passes the check, the entire system is concluded to be internally stable.

\begin{figure}[t]
   \centering
   \includegraphics{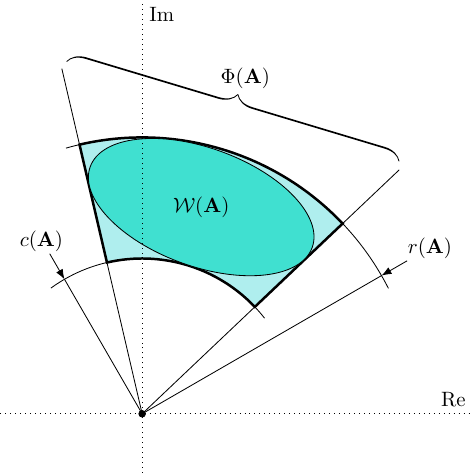}
   \caption{Ring-sectorial bounds of the numerical range $\calW(\bfA)$.}
   \label{fig: ring-sectorial bounds}
\end{figure}

\section{Main Results}
\label{se: main}

Our goal is to extend the compositional approach of Section~\ref{se: compositional small gain} from the small-gain test~\eqref{eq: small gain} to the entire small-gain-or-phase Theorem~\ref{th: small gain phase}. This calls for a phase sector analogue of Theorem~\ref{th: norm compositions}. Since the phase sector $\Phi(\cdot)$ is not composable on its own, a composable set $\calS(\cdot) \ni \Phi(\cdot)$ is required. We propose the following set of 6 functions:
\begin{equation}
   \calS \colon \bfA \mapsto \calS(\bfA) =
   \begin{Bmatrix}
      r(\bfA) & r(\bfA\inv) \\
      c(\bfA) & c(\bfA\inv) \\
      \Phi(\bfA) & \Phi(\bfA\inv)
   \end{Bmatrix}
\end{equation}
defined for an invertible $\bfA$.

To bound $\calS(\bfA)$ means to bound $r(\cdot)$, $c(\cdot)$, and $\Phi(\cdot)$ as in~\eqref{eq: conservative bounds}, i.e., to obtain (conservative) ring-sectorial bounds on both $\calW(\bfA)$ and $\calW(\bfA\inv)$.

Composability of the $\calS(\cdot)$-bounds is established in the current section. For completeness, however, let us first formulate the proposed stability analysis algorithm.

\subsection{Compositional Small-Gain-or-Phase Analysis}
\label{se: compositional small gain phase}

Given a decomposable system in Problem~\ref{problem}, decompose it into a hierarchy of elementary interconnections. Suppose that $\norm{\cdot}$-bounds and, optionally\footnote{For example, if $\bfA$ is not invertible, then $\calS(\bfA)$ is undefined.}, $\calS(\cdot)$-bounds are available for every bottom-level subsystem at each frequency $\omega \in [0, \infty]$.

Starting from the bottom level, perform two steps at each consecutive level of the interconnection hierarchy:
\begin{enumerate}
   \item \emph{Check stability.} Apply Theorem~\ref{th: small gain phase} to verify that every elementary feedback loop at the current level is internally stable. The small-phase part of the theorem is only relevant if $\calS(\cdot)$-bounds are available for both subsystems.
   \item \emph{Compose the bounds.} Apply Theorem~\ref{th: norm compositions} and the forthcoming results of this section to attain $\norm{\cdot}$-bounds and, if possible, $\calS(\cdot)$-bounds for each composite system formed at the current level, so that Theorem~\ref{th: small gain phase} can be used again at the next level.
\end{enumerate}
If every stability check succeeds, internal stability of the entire system is concluded.

\subsection{Crawford Defect}
\label{se: Crawford defect}

In the rest of this Section~\ref{se: main} we present formulas for the compositions of $\calS(\cdot)$-bounds. Therein, ring-sectorial bounds on $\calW(\bfA)$ and $\calW(\bfA\inv)$ included in $\calS(\bfA)$ are coupled through the following quantity.

\begin{definition}[Crawford defect]
   \label{de: crawford defect}
   Let $\bfA \in \bbC^{n\times n}$ be sectorial. Then
   \begin{equation}
      \delta(\bfA) = \sqrt{\frac{1}{c(\bfA) c(\bfA\inv)} - 1} \geq 0
   \end{equation}
   is called the \emph{Crawford defect} of $\bfA$.
\end{definition}

Note that $\calS(\cdot)$-bounds imply an upper bound on $\delta(\cdot)$.

\begin{remark}[Crawford defect as spread]
   \label{re: defect as spread}
   By Lemma~\ref{le: inverse bounds}
   \begin{equation}
      \label{eq: defect as spread}
      \sqrt{\frac{\norm{\bfA}}{c(\bfA)} - 1}
      \leq \delta(\bfA) \leq
      \sqrt{\p*{\frac{\norm{\bfA}}{c(\bfA)}}^2 - 1}
   \end{equation}
   which shows that $\delta(\bfA)$ controls the ratio $\norm{\bfA}/c(\bfA)$ and can be interpreted as a measure of the \textit{relative radial spread} of $\calW(\bfA)$. In particular, $\delta(\bfA) = 0$ if and only if $\bfA$ is scalar.
\end{remark}

\begin{example}[Near-identity matrix]
   \label{ex: near identity}
   Let $\bfA = \rmI + \bfDelta$ where $\norm{\bfDelta} \leq \epsilon < 1$. Then $\bfA$ admits the $\calS(\cdot)$-bounds
   \begin{subequations}
      \begin{align}
         r(\bfA) &\leq 1 + \epsilon, \\
         \label{eq: c near identity}
         c(\bfA) &\geq 1 - \epsilon, \\
         \label{eq: Phi near identity}
         \Phi(\bfA) &\subseteq \arc[-\arcsin\epsilon, \arcsin\epsilon], \\
         \label{eq: rinv near identity}
         r(\bfA\inv) &\leq \frac{1}{1 - \epsilon}, \\
         \label{eq: cinv near identity}
         c(\bfA\inv) &\geq \frac{1}{1 + \epsilon}, \\
         \label{eq: Phiinv near identity}
         \Phi(\bfA\inv) &\subseteq \arc[-\arcsin\epsilon, \arcsin\epsilon], \\
         \intertext{and consequently}
         \delta(\bfA) &\leq \sqrt{\frac{2\epsilon}{1-\epsilon}}.
      \end{align}
   \end{subequations}
   Here~\eqref{eq: rinv near identity} and~\eqref{eq: Phiinv near identity} follow from~\eqref{eq: c near identity} and~\eqref{eq: Phi near identity} via Lemma~\ref{le: inverse bounds}. For the proof of~\eqref{eq: cinv near identity} see Appendix~\ref{ap: proof cinv near identity}.
\end{example}

\subsection{Series Connection}
\label{se: series}

The composition of $\calS(\cdot)$-bounds under the series interconnection (Fig.~\ref{fig: connection series}) is given in the following theorem.

\begin{theorem}
   \label{th: series}
   Let $\bfA_1, \bfA_2, \dots, \bfA_N \in \bbC^{n \times n}$ be sectorial matrices with known bounds on each $\calS(\bfA_k)$~-- in particular,
   \begin{subequations}
      \label{eq: given bounds}
      \begin{align}
         r(\bfA_k) &\leq \hat r_k, \\
         c(\bfA_k) &\geq \hat c_k, \\
         \Phi(\bfA_k) &\subseteq \hat\Phi_k, \\
         \delta(\bfA_k) &\leq \hat\delta_k
      \end{align}
   \end{subequations}
   where $\hat\delta_k$ follows from the bounds on $c(\bfA_k)$ and $c(\bfA_k\inv)$. Define
   \begin{subequations}
      \begin{align}
         \bfB &= \bfA_N \bfA_{N-1} \dots \bfA_1, \\
         \bfD_k &= \begin{bmatrix}
            1 & \hat\delta_k \\ \hat\delta_k & 1 + \hat\delta^2_k
         \end{bmatrix}, \\
         \label{eq: series defect}
         d &= \begin{bmatrix}
            1 & 0
         \end{bmatrix} \bfD_N \bfD_{N-1} \dots \bfD_1 \begin{bmatrix}
            1 \\ 0
         \end{bmatrix} - 1.
      \end{align}
   \end{subequations}
   If $d < 1$ then $\bfB$ is sectorial and
   \begin{subequations}
      \label{eq: bounds series}
      \begin{align}
         r(\bfB) &\leq (1+d) \prod_{k=1}^N \hat r_k, \\
         c(\bfB) &\geq (1-d) \prod_{k=1}^N \hat c_k, \\
         \Phi(\bfB) &\subseteq \sum_{k=1}^N \hat\Phi_k + \arc[-\arcsin d, \arcsin d].
      \end{align}
   \end{subequations}
   Ring-sectorial bounds on $\calW(\bfB\inv)$ are given by this theorem applied to the matrices $\bfA_1\inv, \bfA_2\inv, \dots, \bfA_N\inv$, thus completing the bounds on $\calS(\bfB)$.
\end{theorem}

\begin{IEEEproof}
   See Appendix~\ref{ap: proof series}.
\end{IEEEproof}

\begin{remark}[Two-matrix case]
   \label{re: two matrix series}
   If $N = 2$, in~\eqref{eq: series defect} we have $d = \hat\delta_1 \hat\delta_2$, and Theorem~\ref{th: series} implies
   \begin{equation}
      \delta(\bfA_1 \bfA_2) \leq \dfrac{\hat\delta_1 + \hat\delta_2}
      {1 - \hat\delta_1 \hat\delta_2}
   \end{equation}
   which can be reformulated as
   \begin{equation}
      \label{eq: arctan delta summed}
      \arctan\delta(\bfA_1 \bfA_2) \leq \arctan\hat\delta_1 + \arctan\hat\delta_2
   \end{equation}
   and hints at a certain meaning behind the angle $\arctan\delta(\cdot)$. See also Section~\ref{se: normalized range} on its connection to the \emph{segmental phase}.
\end{remark}

\subsection{Parallel Connection}
\label{se: parallel}

The composition of $\calS(\cdot)$-bounds under the parallel interconnection (Fig.~\ref{fig: connection parallel}) is given in the following theorem.

\begin{theorem}
   \label{th: parallel}
   Let $\bfA_1, \bfA_2, \dots, \bfA_N \in \bbC^{n \times n}$ be sectorial\footnote{See Remark~\ref{re: non sectorial inputs to parallel} if some $\bfA_k$ are not sectorial.} matrices with known bounds~\eqref{eq: given bounds}. Define
   \begin{subequations}
      \begin{align}
         \bfB &= \bfA_1 + \bfA_2 + \dots + \bfA_N, \\
         \hat\calR &= \set[\bigg]{
            \sum_{k=1}^N z_k \colon
            \abs{z_k} \in [\hat c_k, \hat r_k],
            \arg z_k \in \hat\Phi_k
         }, \\
         \hat r &= \max\set[\big]{
            \abs{z} \colon z \in \hat\calR
         }, \\
         \hat c &= \min\set[\big]{
            \abs{z} \colon z \in \hat\calR
         }, \\
         \hat\Phi &= \set{
            \arg z \colon z \in \hat\calR
         }.
      \end{align}
   \end{subequations}
   Then
   \begin{subequations}
      \label{eq: parallel bounds}
      \begin{align}
         r(\bfB) &\leq \hat r, \\
         c(\bfB) &\geq \hat c, \\
         \Phi(\bfB) &\subseteq \hat\Phi.
         \intertext{Furthermore, if $0 \not\in \hat\calR$ then $\bfB$ is sectorial and}
         r(\bfB\inv) &\leq \frac{1}{\hat c}, \\
         \label{eq: parallel bounds c inv}
         c(\bfB\inv) &\geq
            \frac{1}{
               \hat r + \dfrac{1}{\hat c} \p[\Big]{
                  \sum\limits_{k=1}^N \hat\delta_k \hat r_k
               }^2
            }, \\
         \Phi(\bfB\inv) &\subseteq -\hat\Phi.
      \end{align}
   \end{subequations}
\end{theorem}

\begin{IEEEproof}
   Everything except~\eqref{eq: parallel bounds c inv} follows trivially from the subadditivity of the numerical range and Lemma~\ref{le: inverse bounds}. For the proof of~\eqref{eq: parallel bounds c inv}, see Appendix~\ref{ap: proof parallel}.
\end{IEEEproof}

\begin{remark}[Computation in Theorem~\ref{th: parallel}]
   The set $\hat\calR$ is the Minkowski sum of the ring-sectorial bounds on $\calW(\bfA_k)$ which are non-convex, so the geometry here may be complicated. To simplify it, one can inscribe the ring-sectorial bounds into convex polygons and Minkowski-sum those. An even simpler alternative is the approximation
   \begin{subequations}
      \begin{align}
         \hat r &\leq \sum_{k=1}^N \hat r_k, \\
         \hat c &\geq \sum_{k=1}^N \hat c_k \cos\p[\bigg]{
            \tfrac12 \diam\bigvee_{k=1}^N \hat\Phi_k
         }, \\
         \label{eq: parallel sector join}
         \hat\Phi &\subseteq \bigvee_{k=1}^N \hat\Phi_k.
      \end{align}
   \end{subequations}
   Inclusion~\eqref{eq: parallel sector join} conveys the same idea as~\cite[Theorem~7.1]{wangPhasesComplexMatrix2020}.
\end{remark}

\begin{remark}[Non-sectorial inputs to Theorem~\ref{th: parallel}]
   \label{re: non sectorial inputs to parallel}
   The theorem holds if at least some bounds on $\calW(\bfA_k)$ are available, not necessarily the ring-sectorial ones. For example, the norm-disk bound may be sufficient:
   \begin{equation}
      \norm{\bfA_k} \leq \hat\gamma_k \implies
      \calW(\bfA_k) \subseteq \set[\big]{z \colon \abs{z} \leq \hat\gamma_k}.
   \end{equation}
   The set $\hat\calR$ is still obtained as the Minkowski sum of the bounds on $\calW(\bfA_k)$, and bounds~\eqref{eq: parallel bounds} hold except for~\eqref{eq: parallel bounds c inv} which must be replaced with Lemma~\ref{le: inverse bounds}. This differentiates Theorem~\ref{th: parallel} from Theorem~\ref{th: series} in that the resulting matrix $\bfB$ here may turn out to be sectorial even though some of $\bfA_k$ are not. The same can be said about the feedback connection (Section~\ref{se: feedback}).
\end{remark}

\begin{example}[Small additive perturbation]
   \label{ex: bounds under additive perturbation}
   Consider sectorial $\bfA$ and small $\bfDelta$ such that $\norm{\bfDelta} \leq \epsilon < c(\bfA)$. Then
   \begin{subequations}
      \label{eq: bounds under additive perturbation}
      \begin{align}
         r(\bfA + \bfDelta) &\leq r(\bfA) + \epsilon, \\
         c(\bfA + \bfDelta) &\geq c(\bfA) - \epsilon, \\
         \Phi(\bfA + \bfDelta) &\subseteq \Phi(\bfA) + \arc[-\beta, \beta]
      \end{align}
   \end{subequations}
   where $\beta = \arcsin\dfrac{\epsilon}{c(\bfA)}$.
\end{example}

\subsection{Feedback Connection}
\label{se: feedback}

The transfer matrix of the feedback loop in Fig.~\ref{fig: connection feedback} is
\begin{subequations}
   \label{eq: feedback}
   \begin{align}
      \label{eq: feedback long}
      \bfB &= (\rmI + \bfA_1 \bfA_2)\inv \bfA_1 \\
      \label{eq: feedback short}
      &= (\bfA_1\inv + \bfA_2)\inv
   \end{align}
\end{subequations}
where~\eqref{eq: feedback short} is valid whenever $\bfA_1$ is invertible. Representation~\eqref{eq: feedback long} suggests a three-step estimation of $\calS(\bfB)$:
\begin{itemize}
   \item apply Theorem~\ref{th: series} to $\bfA_1$ and $\bfA_2$;
   \item apply Theorem~\ref{th: parallel} to $\rmI$ and $\bfA_1 \bfA_2$;
   \item apply Theorem~\ref{th: series} to $(\rmI + \bfA_1 \bfA_2)\inv$ and $\bfA_1$.
\end{itemize}
Representation~\eqref{eq: feedback short}, on the other hand, suggests the direct approach:
\begin{itemize}
   \item apply Theorem~\ref{th: parallel} to $\bfA_1\inv$ and $\bfA_2$.
\end{itemize}
Generally, neither of the two approaches is superior: they may yield complementary results.

If the feedback loop does not satisfy the small-gain condition, Theorem~\ref{th: norm compositions} does not apply, so we need alternative bounds on $\norm{\bfB}$. From~\eqref{eq: feedback long} and~\eqref{eq: feedback short} using Lemma~\ref{le: inverse bounds} and subadditivity of the numerical range we obtain
\begin{subequations}
   \label{eq: feedback norm bounds}
   \begin{align}
      \label{eq: feedback norm bound long}
      \norm{\bfB} &\leq \frac{\norm{\bfA_1}}{\dist\p[\big]{-1, \calW(\bfA_1 \bfA_2)}}, \\
      \label{eq: feedback norm bound short}
      \norm{\bfB} &\leq \frac{1}{\dist\p[\big]{-\calW(\bfA_1\inv), \calW(\bfA_2)}}.
   \end{align}
\end{subequations}

\begin{remark}
   Bound~\eqref{eq: feedback norm bound short} is analogous to the nonlinear (scaled relative graph) results~\cite[Lemma~5]{chaffeyHomotopyTheoremIncremental2026} and~\cite[Theorem~5]{krebbekxGraphicalAnalysisNonlinear2025}. With $\bfA_1 = \lambda\inv \rmI$ and $\bfA_2 = \bfA$, both variants~\eqref{eq: feedback norm bounds} yield the standard \emph{resolvent bound}~\cite[Theorem~V-3.2]{katoPerturbationTheoryLinear1984}
   \begin{equation}
      \label{eq: resolvent}
      \norm[\big]{(\lambda\rmI + \bfA)\inv} \leq \frac{1}{\dist\p[\big]{{-\lambda}, \calW(\bfA)}}.
   \end{equation}
\end{remark}

The denominators of~\eqref{eq: feedback norm bounds} can be lower-bounded if the bounds on $\calS(\bfA_1)$ and $\calS(\bfA_2)$ are available, but the analytical expressions may not be elegant. Two simplifications are given in the following statement which is complementary to Theorem~\ref{th: norm compositions}, particularly when the small-gain condition of the latter does not hold.

\begin{theorem}[Extension to Theorem~\ref{th: norm compositions}]
   \label{th: feedback norm}
   Let $\bfA_1, \bfA_2 \in \bbC^{n\times n}$ be sectorial matrices. Define $\bfB = (\rmI + \bfA_1 \bfA_2)\inv \bfA_1$ assuming that it exists, e.g., due to the small-gain or small-phase condition.

   If $\theta\mmin = \dist\p[\big]{[\pi] - \Phi(\bfA_1), \Phi(\bfA_2)} > 0$ then
   \begin{equation}
      \label{eq: feedback norm bound phase gap}
      \norm{\bfB} \leq \frac{1}{\max\set[\big]{c(\bfA_1\inv), c(\bfA_2)}
      \sin\min\set[\big]{\theta\mmin, \frac\pi2}}
   \end{equation}
   which may be called a \emph{phase-gap} bound ($\theta\mmin$ is the gap between the phase sectors of $-\bfA_1\inv$ and $\bfA_2$).

   Furthermore, if $c(\bfA_1) c(\bfA_2) > 1$ then
   \begin{equation}
      \label{eq: feedback norm bound high gain}
      \norm{\bfB} \leq \frac{1}{c(\bfA_2) - c(\bfA_1)\inv}
   \end{equation}
   which may be called a \emph{high-gain} bound (it holds if the product $\calW(\bfA_1)\calW(\bfA_2)$ lies further than 1 from zero).
\end{theorem}

\begin{IEEEproof}
   Both bounds follow from~\eqref{eq: feedback norm bound short}: bound~\eqref{eq: feedback norm bound phase gap}~-- via the inclusions
   \begin{subequations}
      \begin{align}
         \notag
         -\calW(\bfA_1\inv) &\subset \set[\big]{z \in \bbC \colon
         \abs{z} \geq c(\bfA_1\inv), \\
         \label{eq: bound of minus inverse range}
         &\phantom{\subset \big\lbrace z \in \bbC \colon} \:
         \arg z \in [\pi] - \Phi(\bfA_1)}, \\
         \notag
         \calW(\bfA_2) &\subset \set[\big]{z \in \bbC \colon
         \abs{z} \geq c(\bfA_2), \\
         &\phantom{\subset \big\lbrace z \in \bbC \colon} \:
         \arg z \in \Phi(\bfA_2)},
      \end{align}
   \end{subequations}
   and~\eqref{eq: feedback norm bound high gain}~-- via
   \begin{subequations}
      \begin{align}
         \label{eq: bound of inverse range}
         \calW(\bfA_1\inv) &\subset \set[\big]{z \in \bbC \colon
         \abs{z} \leq c(\bfA_1)\inv}, \\
         \calW(\bfA_2) &\subset \set[\big]{z \in \bbC \colon
         \abs{z} \geq c(\bfA_2)}
      \end{align}
   \end{subequations}
   where~\eqref{eq: bound of minus inverse range} and~\eqref{eq: bound of inverse range} are justified by Lemma~\ref{le: inverse bounds}.
\end{IEEEproof}

\section{Discussion}
\label{se: discussion}

\subsection{Interplay of the $\norm{\cdot}$-Bounds and $\calS(\cdot)$-Bounds}
\label{se: norm radius interplay}

In the algorithm of Section~\ref{se: compositional small gain phase}, bounds on $\norm{\cdot}$ and $\calS(\cdot)$ are both carried through the iterative composition procedure. The $r(\cdot)$-bound included in $\calS(\cdot)$ may sometimes be improved using the $\norm{\cdot}$-bound: $r(\cdot) \leq \norm{\cdot}$. Conversely, $\calS(\cdot)$-bounds may help improve the $\norm{\cdot}$-bound due to the recently discovered property of sectorial operators~\cite[Theorem~1.1]{druryNumericalRadiusInequality2024}
\begin{equation}
   \label{eq: norm bound via radius}
   \norm{\bfA} \leq r(\bfA) \sqrt{1 + \sin^2\frac{\diam\Phi(\bfA)}{2}}
\end{equation}
where $\tfrac12 \diam\Phi(\bfA)$ is the \emph{sectorial phase radius} of $\bfA$. Ring-sectorial bounds~\eqref{eq: conservative bounds} imply
\begin{equation}
   \label{eq: sectorial radius from ring-sectorial bounds}
   \tfrac12\diam\Phi(\bfA) \leq \min\set*{
      \tfrac12 \diam\hat\Phi, \:
      \arccos\frac{\hat c}{\hat r}
   }.
\end{equation}

Lemma~\ref{le: inverse bounds} contains other useful inequalities with which, e.g., the bound on $c(\bfB\inv)$ in Theorem~\ref{th: parallel} can be corrected to
\begin{equation}
   \label{eq: parallel bound c inv corrected}
   c(\bfB\inv) \geq \max\set*{
      \frac{\hat c}{\hat \gamma^2}, \:
      \frac{1}{\hat r + \dfrac{1}{\hat c} \p[\Big]{
         \sum\limits_{k=1}^N \hat\delta_k \hat r_k
      }^2}}
\end{equation}
where the norm bound $\hat\gamma \geq \norm{\bfB}$ can be taken from Theorem~\ref{th: norm compositions} or from~\eqref{eq: norm bound via radius}, whichever is smaller. For brevity, we did not include all such options in Theorems~\ref{th: series} and~\ref{th: parallel}.

\subsection{How Tight is Theorem~\ref{th: series}?}
\label{se: tightness}

This section is a semi-formal discussion of the \emph{tightness}\footnote{We say that a bound in Theorem~\ref{th: series} is \emph{tight} for a certain class of matrices $\bfA_k$ if there is an example from this class where, given exact bounds on $\calS(\bfA_k)$, the bound holds as an equality.} of the bounds in Theorem~\ref{th: series}. Let us list the takeaways first:
\begin{itemize}
   \item the bound on $c(\bfB)$ is tight for a nontrivial subclass of scalar multiples of unitary matrices;
   \item the bounds on $r(\bfB)$ and $\Phi(\bfB)$ are tight if and only if all matrices $\bfA_k$ except maybe one are scalar;
   \item the bound on $r(\bfB)$ can sometimes be tightened using the norm bounds as explained in Section~\ref{se: norm radius interplay}.
\end{itemize}
Now we shall justify these statements.

If all $\bfA_k$ except maybe one are scalar, all bounds are clearly tight. In the non-scalar case, we limit our attention to a pair of equal matrices: $\bfA_1 = \bfA_2 = \bfA$. Let us start with the following example.

\begin{example}
   \label{ex: tightness}
   Take
   \begin{equation}
      \bfA = \begin{bmatrix}
         \ee^{j\epsilon} & 0 \\
         0 & \ee^{-j\epsilon}
      \end{bmatrix}, \quad 0 < \epsilon < \frac{\pi}{4}
   \end{equation}
   and consider the bounds on $\calS(\bfA^2)$ derived by Theorem~\ref{th: series} from the bounds on $\calS(\bfA)$. We have
   \begin{equation}
      \calW(\bfA) = \calW(\bfA\inv) = \conv\set{\ee^{j\epsilon}, \ee^{-j\epsilon}}
   \end{equation}
   and thus
   \begin{subequations}
      \begin{align}
         r(\bfA) &= 1, \\
         c(\bfA) &= \cos\epsilon, \\
         \Phi(\bfA) &= \arc[-\epsilon, \epsilon], \\
         \delta(\bfA) &= \tan\epsilon.
      \end{align}
   \end{subequations}
   The lower bound on $c(\bfA^2)$ given by Theorem~\ref{th: series} is tight:
   \begin{equation}
      \cos 2\epsilon = \underbrace{c(\bfA^2)
      \geq \p[\big]{1 - \delta(\bfA)^2} c(\bfA)^2}_\text{Theorem~\ref{th: series}}
      = \cos 2\epsilon.
   \end{equation}
   The upper bound on $r(\bfA^2)$ is loose:
   \begin{equation}
      \label{eq: radius bound unitary example}
      1 = \underbrace{r(\bfA^2)
      \leq \p[\big]{1 + \delta(\bfA)^2} r(\bfA)^2}_\text{Theorem~\ref{th: series}}
      = \frac{1}{\cos^2\epsilon}.
   \end{equation}
   The outer bound on $\Phi(\bfA^2)$ is also loose:
   \begin{multline}
      \label{eq: phase bound unitary example}
      \arc[-2\epsilon, 2\epsilon] = \underbrace{\Phi(\bfA^2)
      \subseteq 2\Phi(\bfA) + \arc[-\beta, \beta]}_\text{Theorem~\ref{th: series}} \\
      = \arc[-2\epsilon - \beta, 2\epsilon + \beta]
   \end{multline}
   where $\beta = \arcsin\p[\big]{\delta(\bfA)^2} = \arcsin \tan^2\epsilon$.
\end{example}

How representative is Example~\ref{ex: tightness} of the general situation?

The proof of Theorem~\ref{th: series} suggests that the bound on $c(\bfA^2)$ is tight only if (but not \emph{if and only if}) the inequality in Lemma~\ref{le: orthogonal bound} holds as equality which requires
\begin{equation}
   \label{eq: simultaneous minimum}
   \frac{\abs{\bm\xi^* \bfA \bm\xi}}{\norm{\bm\xi}^2} = c(\bfA), \quad
   \frac{\abs{\bm\xi^* \bfA \bm\xi}}{\norm{\bfA\bm\xi}^2} = c(\bfA\inv)
\end{equation}
with some $\bm\xi \in \bbC^n$. Conditions~\eqref{eq: simultaneous minimum} hold if $\bfA$ is a scalar multiple of unitary, such as in Example~\ref{ex: tightness}.

For the bound on $r(\bfA^2)$ to be tight it is necessary that, in addition to~\eqref{eq: simultaneous minimum}, $\dfrac{\abs{\bm\xi^* \bfA \bm\xi}}{\norm{\bm\xi}^2} = r(\bfA)$ which is incompatible with~\eqref{eq: simultaneous minimum} unless $\bfA$ is scalar.

Finally, for the bound on $\Phi(\bfA^2)$ to be tight it is required that, in addition to~\eqref{eq: simultaneous minimum}, the point $\zeta = \dfrac{\bm\xi^* \bfA \bm\xi}{\norm{\bm\xi}^2} \in \calW(\bfA)$ be such that $\arg\zeta$ lies on the boundary of $\Phi(\bfA)$. Together with~\eqref{eq: simultaneous minimum} it requires that $\zeta$ lie in a corner of the ring-sectorial bounds of $\calW(\bfA)$ and thus be a corner point of the boundary of $\calW(\bfA)$. By~\cite[Theorem~1]{donoghueNumericalRangeBounded1957}, such $\zeta$ is an eigenvalue of $\bfA$ with $\bm\xi$ as its eigenvector. Then from~\eqref{eq: simultaneous minimum}
\begin{equation}
   c(\bfA) = \frac{\abs{\bm\xi^* \bfA \bm\xi}}{\norm{\bm\xi}^2}
   = \abs{\zeta}, \quad
   c(\bfA\inv) = \frac{\abs{\bm\xi^* \bfA \bm\xi}}{\norm{\bfA\bm\xi}^2}
   = \frac{1}{\abs{\zeta}}
\end{equation}
and $\delta(\bfA) = 0$. Therefore, $\bfA$ must be scalar.

Recalling the ideas of Section~\ref{se: norm radius interplay}, one may wonder if $\norm{\bfA}$ could yield a better bound on $r(\bfA^2)$. Indeed, in Example~\ref{ex: tightness} the bound $r(\bfA^2) \leq \norm{\bfA}^2 = 1$ is tight. If the goal is to bound $r(\bfA^2)$ using only the bounds on $\calS(\bfA)$, we can take the route
\begin{multline}
   \label{eq: example radius via norm}
   r(\bfA^2) \leq \norm{\bfA}^2 \overset{\eqref{eq: norm bound via radius}}{\leq}
   r(\bfA)^2 \p*{1 + \sin^2 \frac{\diam\Phi(\bfA)}{2}} \\
   = 1 + \sin^2 \epsilon
\end{multline}
which is loose but still better than~\eqref{eq: radius bound unitary example}. An example where bounding $r(\bfA^2)$ via~\eqref{eq: example radius via norm} is tight is a positive-definite Hermitian $\bfA$: then $\Phi(\bfA) = [0]$ and thus $r(\bfA^2) \leq r(\bfA)^2$ which can be a dramatic improvement over Theorem~\ref{th: series} if $\delta(\bfA)$ is large (i.e., if $\bfA$ has widely spread eigenvalues). In general, however, norm-based bounds are not always better than Theorem~\ref{th: series}: consider $\bfA_1$ with $r(\bfA_1) < \norm{\bfA_1}$ and a scalar $\bfA_2$~-- then the bound on $r(\bfA_2\bfA_1)$ by Theorem~\ref{th: series} is tight whereas the norm-based bound is not.

\subsection{Normalized Numerical Range and Segmental Phase}
\label{se: normalized range}

Through the arguments of Section~\ref{se: tightness} one may come to the conclusion that a major source of conservatism in Theorem~\ref{th: series} is the bound established by Lemma~\ref{le: orthogonal bound} in Appendix~\ref{ap: lemma}. Let us discuss a possible way to reduce the conservatism.

Define the \emph{normalized numerical range}~\cite{auzingerSectorialOperatorsNormalized2003}
\begin{equation}
   \calW_{\rm n}(\bfA) = \min\set*{
      \frac{\abs{\bfx^* \bfA \bfx}}{\norm{\bfx}\norm{\bfA\bfx}} \colon
      \bfx \in \bbC^n \setminus \set{0}, \bfA\bfx \neq 0
   },
\end{equation}
\emph{normalized Crawford number}
\begin{equation}
   \label{eq: normalized Crawford}
   c_{\rm n}(\bfA) = \min\set[\big]{
      \abs{z} \colon z \in \calW_{\rm n}(\bfA)
   }
\end{equation}
and, assuming $c_{\rm n}(\bfA) > 0$, \emph{normalized Crawford defect}
\begin{equation}
   \delta_{\rm n}(\bfA)
   = \sqrt{\frac{1}{c_{\rm n}(\bfA) c_{\rm n}(\bfA\inv)} - 1}
   = \sqrt{\frac{1}{c_{\rm n}(\bfA)^2} - 1}.
\end{equation}
Assumption $c_{\rm n}(\bfA) > 0$ holds, in particular, if $\bfA$ is sectorial.

\begin{proposition}
   \label{pr: normalized defect}
   The bounds in Theorems~\ref{th: series} and~\ref{th: parallel} remain valid with $\delta$ replaced by $\delta_{\rm n}$. Furthermore, $\delta_{\rm n}(\cdot) \leq \delta(\cdot)$, so the replacement can only reduce the conservatism. The case $\delta_{\rm n}(\bfA) = \delta(\bfA)$ occurs when $\bfA$ is a scalar multiple of unitary.
\end{proposition}

\begin{IEEEproof}
   See Appendix~\ref{ap: proof normalized defect}.
\end{IEEEproof}

\begin{figure}
   \centering
   \includegraphics{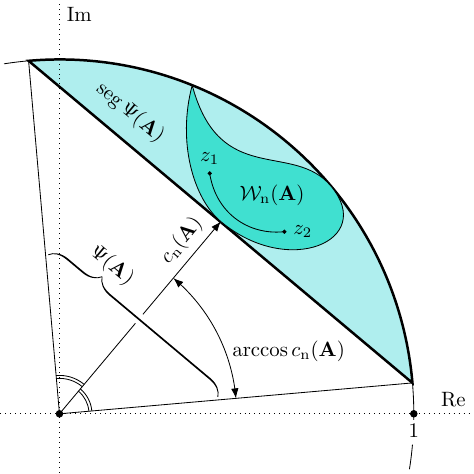}
   \caption{Normalized numerical range $\calW_{\rm n}(\bfA)$ of a sectorial matrix $\bfA$, its normalized Crawford number $c_{\rm n}(\bfA)$, and segmental phase $\Psi(\bfA)$. Points $z_1, z_2 \in \calW_{\rm n}(\bfA)$ connected by a hyperbolic arc illustrate the proof of Lemma~\ref{le: segmental radius}.}
   \label{fig: normalized range}
\end{figure}

The forthcoming Lemma~\ref{le: segmental radius} relates $\delta_{\rm n}(\bfA)$ to the \emph{segmental phase} $\Psi(\bfA)$ which is defined in~\cite{chenCyclicSmallPhase2026} as the minimal phase arc whose chord segment contains $\calW_{\rm n}(\bfA)$:
\begin{equation}
   \label{eq: segmental bound}
   \calW_{\rm n}(\bfA) \subseteq \seg\Psi(\bfA)
\end{equation}
where (see Fig.~\ref{fig: normalized range})
\begin{equation}
   \seg\Psi(\bfA) = \conv\set[\big]{\ee^{jx} \colon [x] \in \Psi(\bfA)}.
\end{equation}
Minimality of the segment can be expressed as
\begin{equation}
   \label{eq: minimal segment}
   \tfrac12\len\Psi(\bfA) = \min_{\theta\in[0,2\pi]}
   \arccos\min_{z \in \calW_{\rm n}(\bfA)} \real(\ee^{j\theta} z)
\end{equation}
where $\tfrac12\len\Psi(\bfA)$ is the \emph{segmental phase radius}\footnote{Denoted $r^\star$ in~\cite{chenCyclicSmallPhase2026} and called the \emph{amplitude} $\operatorname{am}\bfA$ in~\cite{kreinAngularLocalizationSpectrum1969}.} of $\bfA$.

\begin{lemma}
   \label{le: segmental radius}
   If $c_{\rm n}(\bfA) > 0$,
   \begin{equation}
      \label{eq: segmental radius}
      \tfrac12\len\Psi(\bfA) = \arccos c_{\rm n}(\bfA) = \arctan \delta_{\rm n}(\bfA).
   \end{equation}
\end{lemma}

\begin{IEEEproof}
   By~\cite[Lemma~2.5]{linsNormalizedNumericalRange2018}, $\calW_{\rm n}(\bfA)$ is ``convex as seen from the origin'': precisely, every pair $z_1, z_2 \in \calW_{\rm n}(\bfA)$ can be connected by an arc of a hyperbola centered at $0$ and contained in $\calW_{\rm n}(\bfA)$. Optimization problems in~\eqref{eq: normalized Crawford} and~\eqref{eq: minimal segment} are thus convex, dual, and have the same optimal value. See Fig.~\ref{fig: normalized range} for an illustration or Appendix~\ref{ap: proof segmental radius} for a formal proof.
\end{IEEEproof}

By Proposition~\ref{pr: normalized defect} and Lemma~\ref{le: segmental radius}, in the sectorial case
\begin{equation}
   \tfrac12\len\Psi(\bfA) \leq \arctan\delta(\bfA).
\end{equation}
Recall that the angle $\arctan\delta(\cdot)$ is subadditive with respect to matrix multiplication due to Remark~\ref{re: two matrix series}:
\begin{equation}
   \arctan\delta(\bfA_N \bfA_{N-1} \dots \bfA_1)
   \leq \sum_{k=1}^N \arctan\delta(\bfA_k).
\end{equation}
The same property holds for $\tfrac12\len\Psi(\cdot)$~\cite{kreinAngularLocalizationSpectrum1969}:
\begin{equation}
   \label{eq: segmental radius sum}
   \tfrac12\len\Psi(\bfA_N \bfA_{N-1} \dots \bfA_1)
   \leq \sum_{k=1}^{N} \tfrac12\len\Psi(\bfA_k).
\end{equation}
A natural conclusion seems to be that in our results $\delta(\cdot)$ is nothing more than an upper bound on $\delta_{\rm n}(\cdot) = \tan\p[\big]{\tfrac12\len\Psi(\cdot)}$. Had we a more direct bound on $\Psi(\cdot)$, it could be used to improve Theorems~\ref{th: series} and~\ref{th: parallel}. This suggests an alternative approach: instead of bounding $\calS(\cdot)$, it may be possible to develop a composable set of bounds on $\calW(\cdot)$ and $\calW_{\rm n}(\cdot)$.

For example, assuming that all involved matrices are sectorial, consider ring-sectorial bounds on $\calW(\cdot)$ and \emph{segmental bounds}~\eqref{eq: segmental bound} on $\calW_{\rm n}(\cdot)$. Ring-sectorial bounds on $\calW(\cdot)$ are still composable via Theorems~\ref{th: series} and~\ref{th: parallel} with $\delta(\cdot)$ replaced by $\delta_{\rm n}(\cdot) = \tan\p[\big]{\tfrac12\len\Psi(\cdot)}$ and $\Psi(\cdot)$ taken from~\eqref{eq: segmental bound}. As for the segmental bounds on $\calW_{\rm n}(\cdot)$, it is known that segmental phases are \emph{not} composable\footnote{The author thanks Prof. Li Qiu for pointing this out.}: in general,
\begin{equation}
   \label{eq: segmental phase sum wrong}
   \Psi(\bfA_N \bfA_{N-1} \dots \bfA_1) \not\subseteq \sum_{k=1}^N \Psi(\bfA_k).
\end{equation}
The caveat is that $\Psi(\cdot)$ is defined as the arc of the \emph{minimal} chord segment, but the arc on the right-hand side of~\eqref{eq: segmental phase sum wrong} is not necessarily minimal. However, if we view the segmental bounds~\eqref{eq: segmental bound} more loosely, without imposing minimality\footnote{This more general point of view also appears in~\cite[Section~VI]{chenCyclicSmallPhase2026} where it leads to the definition of a \emph{$\gamma$-segmental phase}.}, then they are indeed composable in the following sense.

\begin{proposition}[Segmental bound compositions]
   \label{pr: segmental compositions}
   Let matrices $\bfA_1, \bfA_2, \dots, \bfA_N \in \bbC^{n\times n}$ have known segmental bounds
   \begin{equation}
      \calW_{\rm n}(\bfA_k) \subseteq \seg\Psi_k, \quad
      \Psi_k \subseteq \bbT, \quad k = 1, 2, \dots, N.
   \end{equation}
   Then
   \begin{subequations}
      \begin{align}
         \label{eq: segmental phase summed}
         \calW_{\rm n}(\bfA_N \bfA_{N-1} \dots \bfA_1) &\subseteq
         \seg \sum_{k=1}^N \Psi_k, \\
         \label{eq: segmental phase joined}
         \calW_{\rm n}(\bfA_1 + \bfA_2 + \dots + \bfA_N) &\subseteq
         \seg \bigvee_{k=1}^N \Psi_k
      \end{align}
   \end{subequations}
   and $\calW_{\rm n}\p[\big]{(\rmI + \bfA_1\bfA_2)\inv \bfA_1}$ can be bounded using~\eqref{eq: segmental phase joined} and the relation $\calW_{\rm n}(\bfA\inv) = \calW_{\rm n}(\bfA)^*$.
\end{proposition}

\begin{IEEEproof}
   Bounds~\eqref{eq: segmental phase summed} can be established along the lines of the proof of~\cite[Theorem~1]{chenCyclicSmallPhase2026}, and~\eqref{eq: segmental phase joined} follows from the triangle inequality.
\end{IEEEproof}

A shortcoming of Proposition~\ref{pr: segmental compositions} is that~\eqref{eq: segmental phase joined} is loose even for scalar matrices. We conjecture that~\eqref{eq: segmental phase joined} can be tightened by supplementing $\Psi_k$ with ring-sectorial bounds on $\calW(\bfA_k)$.

Let us conclude this section by reiterating that the above hybrid sectorial-segmental idea is still restricted to sectorial matrices.

\section{Examples}
\label{se: examples}

The compositional algorithm of Section~\ref{se: compositional small gain phase} can be unwieldy in dealing with deep hierarchies of elementary interconnections. Generally, one can expect to attain only \emph{procedural} (as opposed to \emph{declarative}) stability conditions. Procedural conditions must be checked numerically going through the composition process. As usual, conservatism may be reduced to some extent by clever block diagram transformations such as \emph{loop shifting}~\cite[Theorem~6.6.22]{vidyasagarNonlinearSystemsAnalysis1993}, \emph{scaling}~\cite[Theorem~11.3]{zhouRobustOptimalControl1996}, and \emph{node splitting} (Fig.~\ref{fig: node splitting} and Remark~\ref{re: defect as spread} above). We postpone such examples until a future publication.

In the present section, we consider ``shallow'' systems depicted in Figs.~\ref{fig: example 1}, \ref{fig: example 2}, and~\ref{fig: example 3} where simple \emph{declarative} stability conditions are possible (Propositions~\ref{pr: example 1 result}, \ref{pr: example 2 result}, and~\ref{pr: example 3 result}).

\subsection{Perturbation of a Single Loop}
\label{se: example 1}

\begin{figure}
   \centering
   \begin{subfigure}[b]{30mm}
      \centering
      \includegraphics{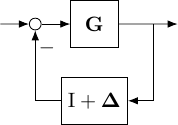}
      \caption{}
      \label{fig: example 1 single loop}
   \end{subfigure}
   \begin{subfigure}[b]{50mm}
      \centering
      \includegraphics{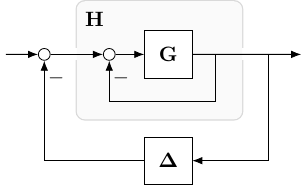}
      \caption{}
      \label{fig: example 1 double loop}
   \end{subfigure}
   \caption{Example of Section~\ref{se: example 1} as a single loop (a) and nested loops (b) prepared for the compositional analysis.}
   \label{fig: example 1}
\end{figure}

Consider the loop in Fig.~\ref{fig: example 1 single loop}: plant $\bfG$ with feedback $\rmI + \bfDelta$ where $\norm{\bfDelta(j\omega)} \leq \epsilon < 1$ at each $\omega \in [0,\infty]$. We assume that bounds on $\norm{\bfG}$ and $\calS(\bfG)$ are available and take bounds on $\calS(\rmI + \bfDelta)$ from Example~\ref{ex: near identity}.

The standard small-gain-or-phase Theorem~\ref{th: small gain phase} asserts that the system in Fig.~\ref{fig: example 1 single loop} is internally stable if at every $\omega \in [0,\infty]$
\begin{subequations}
   \label{eq: example 1 small gain phase}
   \begin{align}
      & \norm{\bfG} < \frac{1}{1+\epsilon} \\
      \text{or}\quad & \dist\p[\big]{[\pi], \Phi(\bfG)} > \arcsin\epsilon
   \end{align}
\end{subequations}
where $\bfG = \bfG(j\omega)$. Graphically, conditions~\eqref{eq: example 1 small gain phase} require that the wedge-shaped bounds on $\calW(\bfG)$ specified by $\norm{\bfG}$ and $\Phi(\bfG)$ remain strictly inside the region seen in Fig.~\ref{fig: example 1 smaller region}.

Aiming to relax~\eqref{eq: example 1 small gain phase}, we put the system into the nested loop form (Fig.~\ref{fig: example 1 double loop}) and apply compositional analysis.

We start by looking at the inner (``nominal'') loop. It is internally stable if \eqref{eq: example 1 small gain phase} holds with $\epsilon = 0$.

Next we proceed to the outer loop. Since $\bfDelta$ has no phase bounds, we impose the pure small-gain condition
\begin{equation}
   \label{eq: example 1 outer small gain}
   \norm{\bfH} < \frac{1}{\epsilon}
\end{equation}
where $\bfH = (\rmI + \bfG)\inv\bfG$ is the inner loop's transfer matrix. Bounds on $\norm{\bfH}$ are found in Theorems~\ref{th: norm compositions} and~\ref{th: feedback norm}:
\begin{subequations}
   \begin{align}
      \norm{\bfH} &\leq \frac{1}{\norm{\bfG}\inv - 1}, \\
      \norm{\bfH} &\leq \frac{1}
      {\sin\min\set[\big]{\dist\p[\big]{[\pi] - \Phi(\bfG), [0]},
      \tfrac\pi2}}, \\
      \norm{\bfH} &\leq \frac{1}{1 - c(\bfG)\inv}.
   \end{align}
\end{subequations}
These bounds, respectively, suggest the following alternatives, each of them sufficient for~\eqref{eq: example 1 outer small gain}:
\begin{subequations}
   \label{eq: example 1 outer loop conditions}
   \begin{align}
      &\norm{\bfG} < \frac{1}{1 + \epsilon} \\
      \text{or}\quad &\dist\p[\big]{[\pi], \Phi(\bfG)} > \arcsin\epsilon \\
      \text{or}\quad &c(\bfG) > \frac{1}{1 - \epsilon}.
   \end{align}
\end{subequations}
We conclude that the system is internally stable if, firstly, \eqref{eq: example 1 small gain phase} holds with $\epsilon = 0$ and, secondly, at least one of the options in~\eqref{eq: example 1 outer loop conditions} holds. This results in the following statement.

\begin{proposition}
   \label{pr: example 1 result}
   Let $\bfG, \bfDelta \in \RHinfnn$. The system in Fig.~\ref{fig: example 1} is internally stable if at every frequency $\omega \in [0, \infty]$
   \begin{equation}
      \norm{\bfDelta} \leq \epsilon < 1
   \end{equation}
   and at least one of the following conditions holds:
   \begin{subequations}
      \label{eq: example 1 result}
      \begin{align}
         \label{eq: example 1 result small gain}
         & \hphantom{\Bigg\{}
         \norm{\bfG} < \frac{1}{1+\epsilon} \\
         \label{eq: example 1 result small phase}
         \text{or}\quad & \hphantom{\Bigg\{}
         \dist\p[\big]{[\pi], \Phi(\bfG)} > \arcsin\epsilon \\
         \label{eq: example 1 result high gain small phase}
         \text{or}\quad & \begin{cases}
            c(\bfG) > \dfrac{1}{1-\epsilon}, \\
            \dist\p[\big]{[\pi], \Phi(\bfG)} > 0
         \end{cases}
      \end{align}
   \end{subequations}
   where $\bfDelta = \bfDelta(j\omega)$ and $\bfG = \bfG(j\omega)$.
\end{proposition}

Condition~\eqref{eq: example 1 result small gain} is \emph{small-gain}, \eqref{eq: example 1 result small phase} is \emph{small-phase}, and~\eqref{eq: example 1 result high gain small phase} may be called \emph{high-gain-and-small-phase}.

Proposition~\ref{pr: example 1 result} stipulates that the ring-sectorial bounds on $\calW(\bfG)$ formed by $\norm{\bfG}$, $\Phi(\bfG)$, and $c(\bfG)$ remain strictly inside the region shown in Fig.~\ref{fig: example 1 larger region}. This region is larger than the one in Fig.~\ref{fig: example 1 smaller region} due to the high-gain part~\eqref{eq: example 1 result high gain small phase}.

\begin{figure}
   \centering
   \begin{subfigure}[b]{42mm}
      \centering
      \includegraphics{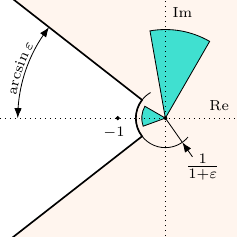}
      \caption{}
      \label{fig: example 1 smaller region}
   \end{subfigure}
   \begin{subfigure}[b]{42mm}
      \centering
      \includegraphics{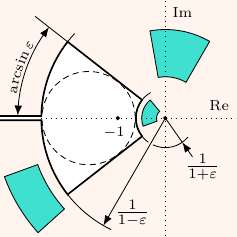}
      \caption{}
      \label{fig: example 1 larger region}
   \end{subfigure}
   \caption{Stability domains in the example of Section~\ref{se: example 1}: (a) illustrates the standard small-gain-or-phase condition~\eqref{eq: example 1 small gain phase}; (b) is the result of compositional analysis (Proposition~\ref{pr: example 1 result}). The turquoise regions are examples of numerical range bounds that satisfy the stability conditions. The dashed circle in (b) is Hall's $M$-circle with $M = \epsilon\inv$ (see Remark~\ref{re: hall}).}
   \label{fig: example 1 regions}
\end{figure}

\begin{remark}[Hall's circle]
   \label{re: hall}
   In the scalar case, condition~\eqref{eq: example 1 outer small gain} is equivalent to $G \in \set*{z \in \bbC \colon \abs*{\dfrac{z}{1 + z}} < \dfrac{1}{\epsilon}}$ and restricts the Nyquist curve to lie outside \emph{Hall's $M$-circle} (with $M = \epsilon\inv$)~\cite[Fig.~27]{hallAnalysisSynthesisLinear1943}. This circle coincides with the dashed circle inscribed in the exclusion region in Fig.~\ref{fig: example 1 larger region}. Thus, to relate our result to the scalar case, the region specified in Proposition~\ref{pr: example 1 result} can be described as: the whole plane, minus the ring-sectorial approximation of the Hall circle, minus the ray $[-\infty, -1]$. The ray is cut out due to the small-phase part of~\eqref{eq: example 1 result high gain small phase} which ensures stability in the nominal case ($\epsilon = 0$).
\end{remark}

\begin{remark}[Unconditional and high-gain stability]
   Suppose that $\bfG$ in Fig.~\ref{fig: example 1 single loop} is replaced with $k\bfG$ where $k \geq 0$. If the loop remains stable for all $k \in [0, 1]$ (and all admissible $\bfDelta$) then it is called (robustly) \emph{unconditionally stable}~\cite{bodeRelationsAttenuationPhase1940} with gain margin $k_* \in (1, \infty]$. Proposition~\ref{pr: example 1 result} restricted to the part ``\eqref{eq: example 1 result small gain} or~\eqref{eq: example 1 result small phase}'' is sufficient for the robust unconditional stability, and in this part compositional analysis yields nothing new compared to the classical condition~\eqref{eq: example 1 small gain phase}. A somewhat opposite property is when the loop remains stable for all $k \geq 1$ and all admissible $\bfDelta$~-- then it is called \emph{robustly high-gain stabilizable} with gain margin $k_* \in [0, 1)$~\cite[Section~5]{tsypkinHighgainRobustControl1999}. A sufficient condition of this is Proposition~\ref{pr: example 1 result} restricted to ``\eqref{eq: example 1 result small phase} or~\eqref{eq: example 1 result high gain small phase}'' which is new compared to~\eqref{eq: example 1 small gain phase}. If Proposition~\ref{pr: example 1 result} holds in its entirety, the system is stable but not necessarily unconditionally stable nor high-gain stabilizable.
\end{remark}

\subsection{Small-Feedback Perturbation of Subsystems}
\label{se: example 2}

\begin{figure}
   \centering
   \includegraphics{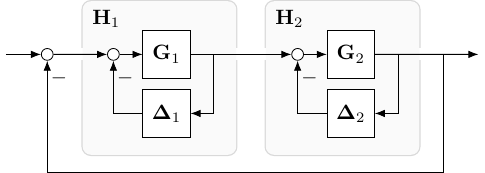}
   \caption{Example of Section~\ref{se: example 2}.}
   \label{fig: example 2}
\end{figure}

Consider the system in Fig.~\ref{fig: example 2}: a loop $\bfG_1 \rightleftharpoons \bfG_2$, but each $\bfG_k$ is perturbed by a small inner feedback $\bfDelta_k$ with $\norm{\bfDelta_k(j\omega)} \leq \epsilon_k$ for all $\omega \in [0,\infty]$.

The inner loops are stable if the small-gain condition
\begin{equation}
   \label{eq: example 2 inner small gain}
   \norm{\bfG_k(j\omega)} < \frac{1}{\epsilon_k}, \quad k = 1,2
\end{equation}
holds at every $\omega\in[0, \infty]$. The outer loop is stable if the small-gain-or-phase Theorem~\ref{th: small gain phase} holds with $\bfA_k = \bfH_k$ where $\bfH_k = (\rmI + \bfG_k \bfDelta_k)\inv \bfG_k$ is the transfer matrix of the $k$'th inner loop. We bound $\norm{\bfH_k}$ and $\Phi(\bfH_k)$ as
\begin{align}
   \label{eq: example 2 inner norm}
   \text{Theorem~\ref{th: norm compositions}}
   &\implies \norm{\bfH_k}
   \leq \frac{1}{\norm{\bfG_k}\inv - \epsilon_k}, \\
   \notag
   \text{Example~\ref{ex: bounds under additive perturbation}}
   &\implies \Phi(\bfH_k)
   = -\Phi\p[\big]{\bfG_k\inv + \bfDelta_k} \\
   \notag
   &\phantom{\implies \Phi(\bfH_k)}\;
   \subseteq -\Phi(\bfG_k\inv) + \arc[-\beta_k, \beta_k] \\
   \label{eq: example 2 inner sector}
   &\phantom{\implies \Phi(\bfH_k)}\;
   = \Phi(\bfG_k) + \arc[-\beta_k, \beta_k], \\
   &\phantom{{} \implies {}}
   \beta_k = \arcsin\dfrac{\epsilon_k}{c(\bfG_k\inv)}.
\end{align}
Plugging~\eqref{eq: example 2 inner norm} and~\eqref{eq: example 2 inner sector} into Theorem~\ref{th: small gain phase} and remembering to ensure~\eqref{eq: example 2 inner small gain} we arrive at the following result.

\begin{proposition}
   \label{pr: example 2 result}
   Let $\bfG_1, \bfG_2, \bfDelta_1, \bfDelta_2 \in \RHinfnn$. The system in Fig.~\ref{fig: example 2} is internally stable if at every $\omega\in[0,\infty]$
   \begin{equation}
      \norm{\bfDelta_k} \leq \epsilon_k, \quad
      \norm{\bfG_k} < \epsilon_k\inv, \quad
      k = 1, 2
   \end{equation}
   and at least one of the following conditions holds:
   \begin{subequations}
      \begin{align}
         \label{eq: example 2 small gain}
         & \p*{\frac{1}{\norm{\bfG_1}} - \epsilon_1}
         \p*{\frac{1}{\norm{\bfG_2}} - \epsilon_2} > 1 \\
         \label{eq: example 2 small phase}
         \text{or}\quad & \dist\p[\big]{[\pi], \Phi(\bfG_1) + \Phi(\bfG_2)}
         > \beta_1 + \beta_2
      \end{align}
   \end{subequations}
   where $\beta_k = \arcsin\dfrac{\epsilon_k}{c(\bfG_k\inv)}$,  $\bfDelta_k = \bfDelta_k(j\omega)$, and $\bfG_k = \bfG_k(j\omega)$. Condition~\eqref{eq: example 2 small phase} is only relevant if both $\bfG_k$ are sectorial and $\epsilon_k < c(\bfG_k\inv)$.
\end{proposition}

Condition~\eqref{eq: example 2 small gain} is \emph{small-gain}, and~\eqref{eq: example 2 small phase} is \emph{small-phase}. In the nominal case ($\epsilon_1 = \epsilon_2 = 0$) Proposition~\ref{pr: example 2 result} reduces to the standard small-gain-or-phase Theorem~\ref{th: small gain phase}.

\subsection{Multiplicative Perturbation of Subsystems}
\label{se: example 3}

\begin{figure}
   \centering
   \includegraphics{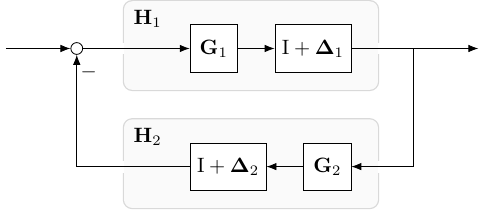}
   \caption{Example of Section~\ref{se: example 3}.}
   \label{fig: example 3}
\end{figure}

Consider the system in Fig.~\ref{fig: example 3}: a loop $\bfG_1 \rightleftharpoons \bfG_2$ affected by multiplicative perturbations $\rmI + \bfDelta_k$ ($\norm{\bfDelta_k} \leq \epsilon_k$). We treat it as a two-layer composition: feedback $\bfH_1 \rightleftharpoons \bfH_2$ where each $\bfH_k$ is the series connection of $\bfG_k$ and $\rmI + \bfDelta_k$.

Applying Theorem~\ref{th: small gain phase} to the loop $\bfH_1 \rightleftharpoons \bfH_2$ with the bounds
\begin{align}
   \text{Theorem~\ref{th: norm compositions}}
   &\implies \norm{\bfH_k} \leq (1 + \epsilon_k) \norm{\bfG_k}, \\
   \text{Theorem~\ref{th: series}}
   &\implies \Phi(\bfH_k) \subseteq \Phi(\bfG_k) + \arc[-\beta_k, \beta_k], \\
   \notag
   &\phantom{{} \implies {}} \beta_k = \arcsin\epsilon_k \\
   &\phantom{{} \implies \beta_k} + \arcsin\p[\big]{\delta(\bfG_k)\delta(\rmI + \bfDelta_k)}, \\
   \label{eq: example 3 defect bound}
   \text{Example~\ref{ex: near identity}}
   &\implies \delta(\rmI + \bfDelta_k) \leq \sqrt{\frac{2\epsilon_k}{1 - \epsilon_k}}
\end{align}
we arrive at the following.

\begin{proposition}
   \label{pr: example 3 result}
   Let $\bfG_1, \bfG_2, \bfDelta_1, \bfDelta_2 \in \RHinfnn$. The system in Fig.~\ref{fig: example 3} is internally stable if at every $\omega\in[0,\infty]$
   \begin{equation}
      \norm{\bfDelta_k(j\omega)} \leq \epsilon_k, \quad
      k = 1, 2
   \end{equation}
   and at least one of the following conditions holds:
   \begin{subequations}
      \begin{align}
         \label{eq: example 3 small gain}
         & \norm{\bfG_1} \norm{\bfG_2} < \frac{1}{(1 + \epsilon_1)(1 + \epsilon_2)} \\
         \label{eq: example 3 small phase}
         \text{or}\quad & \dist\p[\big]{[\pi], \Phi(\bfG_1) + \Phi(\bfG_2)}
         > \beta_1 + \beta_2
      \end{align}
   \end{subequations}
   where $\beta_k = \arcsin\epsilon_k + \arcsin\p*{\delta(\bfG_k) \sqrt{\dfrac{2\epsilon_k}{1 - \epsilon_k}}}$, $\bfDelta_k = \bfDelta_k(j\omega)$, and $\bfG_k = \bfG_k(j\omega)$. Condition~\eqref{eq: example 3 small phase} is only relevant if both $\bfG_k$ are sectorial and $\epsilon_k < \dfrac{1}{1 + 2\delta(\bfG_k)^2}$.
\end{proposition}

Inequality~\eqref{eq: example 3 small gain} is the classical \emph{small-gain} condition, and~\eqref{eq: example 3 small phase} is a new \emph{small-phase} one. In the nominal case ($\epsilon_1 = \epsilon_2 = 0$) Proposition~\ref{pr: example 3 result} reduces to the standard small-gain-or-phase Theorem~\ref{th: small gain phase}.

\begin{remark}[Comparison to small segmental phase]
   Segmental small-phase condition~\cite[Theorem~3]{chenCyclicSmallPhase2026} reads
   \begin{equation}
      \label{eq: example 3 segmental long}
      \dist\p[\big]{[\pi], \Psi(\bfG_1) + \Psi(\rmI + \bfDelta_1) + \Psi(\bfG_2)
      + \Psi(\rmI + \bfDelta_2)} > 0
   \end{equation}
   where $\Psi(\cdot)$ is the segmental phase. Note the inclusion
   \begin{equation}
      \label{eq: example 3 segmental arc bound}
      \Psi(\rmI + \bfDelta_k) \subseteq \arc[-\arcsin\epsilon_k, \arcsin\epsilon_k]
   \end{equation}
   that can be shown by the method of Appendix~\ref{ap: proof cinv near identity} making use of Lemma~\ref{le: segmental radius}. Therefore, \eqref{eq: example 3 segmental long} follows from
   \begin{equation}
      \label{eq: example 3 segmental}
      \dist\p[\big]{[\pi], \Psi(\bfG_1) + \Psi(\bfG_2)}
      > \arcsin\epsilon_1 + \arcsin\epsilon_2
   \end{equation}
   which is a segmental alternative to~\eqref{eq: example 3 small phase}. Even though the right-hand side of~\eqref{eq: example 3 segmental} is smaller than~\eqref{eq: example 3 small phase}, the left-hand side is smaller as well since $\Phi(\cdot) \subseteq \Psi(\cdot)$, so neither~\eqref{eq: example 3 small phase} nor~\eqref{eq: example 3 segmental} generally dominates the other. However, in the non-sectorial case~\eqref{eq: example 3 small phase} is simply inapplicable.
\end{remark}

\begin{remark}[Sectorial-segmental approach]
   Let us try the hybrid sectorial-segmental idea of Section~\ref{se: normalized range} in this example. Recall that Proposition~\ref{pr: example 3 result} is obtained with $\delta(\rmI + \bfDelta_k)$, but by Proposition~\ref{pr: normalized defect} we could use $\delta_{\rm n}(\cdot)$ instead~-- how significant an advantage would it be for small $\epsilon_k$? From~\eqref{eq: example 3 defect bound}
   \begin{equation}
      \delta(\rmI + \bfDelta_k) \leq O(\sqrt{\epsilon_k})
      \quad\text{as } \epsilon_k \to 0
   \end{equation}
   and from~\eqref{eq: example 3 segmental arc bound} by Lemma~\ref{le: segmental radius}
   \begin{equation}
      \delta_{\rm n}(\rmI + \bfDelta_k) \leq \tan\arcsin\epsilon_k
      = O(\epsilon_k) = o(\sqrt{\epsilon_k}),
   \end{equation}
   so the hybrid approach with $\delta_{\rm n}$ is qualitatively sharper.
\end{remark}

\section{Conclusions}

We proposed a way to apply the small-gain-or-phase stability analysis to a class of systems built by a hierarchy of series, parallel, and feedback interconnections. The analysis is based on the composition of the ring-sectorial bounds of numerical ranges. Although the composition formulas can be complicated and conservative, in simple examples they yield analytically manageable and intuitively convincing results.

A potentially attractive feature of the compositional analysis, facilitated by its iterative nature, is the opportunity to recognize the \emph{source} of instability~-- an internal loop that violates the stability conditions. This may prove useful in applications concerned with controller design and system diagnostics.

Possible future advances are discussed in Section~\ref{se: normalized range} (combining the sectorial and segmental approaches) and in Remark~\ref{re: deviation bound} (incorporating Drury's deviation bound or its approximation).

\appendices

\section{A Fundamental Lemma}
\label{ap: lemma}

In the proofs we often adopt a geometric point of view on the numerical range. Note that every point $z \in \calW(\bfA)$ contains some information about the action of $\bfA$ on $\bfx$. Specifically, $\bfA\bfx$ can be decomposed into the scaled component $z\bfx$ and a component orthogonal to $\bfx$. The orthogonal component is discarded from $\calW(\bfA)$, but in the sectorial case it can be bounded proportionally to $\abs{z}$ using the Crawford defect $\delta(\bfA)$~-- this is the essence of the following lemma.

\begin{lemma}
   \label{le: orthogonal bound}
   Let matrix $\bfA \in \bbC^{n\times n}$ be sectorial. For an arbitrary $\bfx \in \bbC^n \setminus \set{0}$, define $z = \dfrac{\bfx^* \bfA \bfx}{\norm{\bfx}^2} \in \calW(\bfA)$ and the orthogonal projection matrix $\bfP = \rmI - \dfrac{\bfx \bfx^*}{\norm{\bfx}^2} \in \bbC^{n\times n}$. Then
   \begin{equation}
      \frac{\norm{\bfP\bfA\bfx}}{\norm{\bfx}} \leq \delta(\bfA) \abs{z}.
   \end{equation}
\end{lemma}

\begin{IEEEproof}
   From the orthogonal decomposition
   \begin{equation}
      \frac{\bfA\bfx}{\norm{\bfx}}
      = \frac{z\bfx}{\norm{\bfx}}
      + \frac{\bfP \bfA \bfx}{\norm{\bfx}}
   \end{equation}
   we find $\p*{\dfrac{\norm{\bfA\bfx}}{\norm{\bfx}}}^2 = \abs{z}^2 + \p*{\dfrac{\norm{\bfP\bfA\bfx}}{\norm{\bfx}}}^2$ which together with
   \begin{equation}
      \label{eq: problematic ratio}
      \p*{\frac{\norm{\bfA\bfx}}{\norm{\bfx}}}^2
      = \frac{\p*{\dfrac{\abs{\bfx^* \bfA \bfx}}{\norm{\bfx}^2}}^2}{
         \dfrac{\abs{\bfx^* \bfA \bfx}}{\norm{\bfx}^2}
         \dfrac{\abs{\bfx^* \bfA \bfx}}{\norm{\bfA\bfx}^2}
      } \leq \frac{\abs{z}^2}{c(\bfA) c(\bfA\inv)}
   \end{equation}
   yields the required estimate.
\end{IEEEproof}

\begin{remark}[Are there no better bounds?]
   \label{re: deviation bound}
   The idea of bounding the $\bfx$-orthogonal component of $\bfA\bfx$ is not new. We particularly note~\cite[Theorem~8]{druryNumericalRangeMatrix2024} where the departure of $\calW(\bfA_1 \bfA_2)$ from $\calW(\bfA_1) \calW(\bfA_2)$ is estimated using a quantity called the \emph{deviation bound}~-- a tight bound on the orthogonal component that does not require sectoriality. At first sight, the deviation bound is difficult to attain whereas our Crawford defect $\delta(\cdot)$ can be easily read off $\calS(\cdot)$. It may be worthwhile to characterize the orthogonal component in a way that is less conservative than $\delta(\cdot)$ (or $\delta_{\rm n}(\cdot)$ from Section~\ref{se: normalized range}) but more computable than the deviation bound.
\end{remark}

\section{Proof of~\eqref{eq: cinv near identity} in Example~\ref{ex: near identity}}
\label{ap: proof cinv near identity}

We need to lower-bound $c(\bfA\inv)$ where $\bfA = \rmI + \bfDelta$ and $\norm{\bfDelta} \leq \epsilon < 1$. Consider the orthogonal decomposition
\begin{equation}
   \bfDelta \bfx = z\bfx + \bfy, \quad z\in\bbC, \quad \bfy\perp\bfx.
\end{equation}
Then $\abs{z} \leq \epsilon$,
\begin{equation}
   \p*{\frac{\norm{\bfy}}{\norm{\bfx}}}^2
   = \p*{\frac{\norm{\bfDelta\bfx}}{\norm{\bfx}}}^2 - \abs{z}^2
   \leq \epsilon^2 - \abs{z}^2,
\end{equation}
and
\begin{subequations}
   \begin{align}
      c(\bfA\inv)
      &= \min_{\bfx\neq 0} \frac{\abs{\bfx^* \bfA \bfx}}
      {\norm{\bfA\bfx}^2} \\
      &= \min_{\bfx\neq 0} \frac{\abs{1+z} \norm{\bfx}^2}
      {\abs{1+z}^2 \norm{\bfx}^2 + \norm{\bfy}^2} \\
      &\geq \min_{\abs{z} \leq \epsilon} \frac{\abs{1+z}}
      {\abs{1+z}^2 - \abs{z}^2 + \epsilon^2} \\
      &\geq \min_{\abs{z} \leq \epsilon} \frac{1 + \real z}
      {1 + 2\real z + \epsilon^2}
      = \frac{1}{1 + \epsilon}.
   \end{align}
\end{subequations}

\section{Proof of Theorem~\ref{th: series}}
\label{ap: proof series}

Take arbitrary $\bfx \in \bbC^n \setminus \set{0}$ and let for $k = 1, 2, \dots, N$
\begin{align}
   z_k &= \frac{\bfx^* \bfA_k \bfx}{\norm{\bfx}^2} \in \calW(\bfA_k), \\
   w_k &= \frac{\bfx^* \bfA_k \bfA_{k-1} \dots \bfA_1 \bfx}{\norm{\bfx}^2},
   \quad w_N \in \calW(\bfB), \\
   \bfy_k &= \frac{\bfP \bfA_k \bfx}{\norm{\bfx}} \perp \bfx, \\
   \tilde\bfy_k^* &= \frac{\bfx^* \bfA_k \bfP}{\norm{\bfx}} \perp \bfx, \\
   \bfv_k &= \frac{\bfP \bfA_k \bfA_{k-1} \dots \bfA_1 \bfx}{\norm{\bfx}} \perp \bfx
\end{align}
where $\bfP$ is the projection matrix from Lemma~\ref{le: orthogonal bound}. Observe the iterative relations
\begin{subequations}
   \label{eq: series iterations}
   \begin{alignat}{3}
      w_k &= z_k w_{k-1} &&+ \tilde\bfy_k^* \bfv_{k-1}^{}, \\
      \bfv_k &= \bfy_k w_{k-1} &&+ \bfP \bfA_k \bfv_{k-1}.
   \end{alignat}
\end{subequations}
We aim to estimate the relative departure of $w_k$ from $z_1 z_2 \dots z_k$. Accordingly, let
\begin{align}
   \omega_k &= \frac{w_k}{z_1 z_2 \dots z_k} - 1, \\
   \bm\upsilon_k &= \frac{\bfv_k}{z_1 z_2 \dots z_k}.
\end{align}
Equations~\eqref{eq: series iterations} turn into
\begin{subequations}
   \label{eq: series iterations normalized}
   \begin{alignat}{3}
      \omega_k &= \omega_{k-1} &&+ \frac{\tilde\bfy_k^*}{z_k}\bm\upsilon_{k-1}, \\
      \bm\upsilon_k &= (\omega_{k-1} + 1) \frac{\bfy_k}{z_k}
      &&+ \frac{\bfP \bfA_k}{z_k} \bm\upsilon_{k-1}.
   \end{alignat}
\end{subequations}
From Lemma~\ref{le: orthogonal bound} and Remark~\ref{re: defect as spread}
\begin{subequations}
   \begin{gather}
      \norm[\Big]{
         \frac{\bfy_k}{z_k}
      } \leq \delta(\bfA_k), \quad
      \norm[\Big]{
         \frac{\tilde\bfy_k}{z_k}
      } \leq \delta(\bfA_k^*) = \delta(\bfA_k), \\
      \norm[\bigg]{
         \frac{\bfP \bfA_k}{z_k}
      } \leq \frac{\norm{\bfA_k}}{c(\bfA_k)}
      \leq 1 + \delta^2(\bfA_k),
   \end{gather}
\end{subequations}
and then from~\eqref{eq: series iterations normalized}
\begin{alignat}{3}
   \abs{\omega_k} &\leq \abs{\omega_{k-1}}
   &&+ \delta(\bfA_k) \norm{\bm\upsilon_{k-1}}, \\
   \norm{\bm\upsilon_k} &\leq \delta(\bfA_k) (\abs{\omega_{k-1}} + 1)
   &&+ \p[\big]{1 + \delta^2(\bfA_k)} \norm{\bm\upsilon_{k-1}}.
\end{alignat}
Thus,
\begin{equation}
   \abs{\omega_k} + 1 \leq a_k, \quad
   \norm{\bm\upsilon_k} \leq b_k
\end{equation}
where $a_k$ and $b_k$ satisfy
\begin{equation}
   \label{eq: series bounds iterations}
   \begin{bmatrix}
      a_k \\ b_k
   \end{bmatrix} = \underbrace{\begin{bmatrix}
      1 & \delta(\bfA_k) \\ \delta(\bfA_k) & 1+\delta^2(\bfA_k)
   \end{bmatrix}}_{\bfD_k} \begin{bmatrix}
      a_{k-1} \\ b_{k-1}
   \end{bmatrix}.
\end{equation}
The initial conditions can be $a_1 = 1$ and $b_1 = \delta(\bfA_1)$ because
\begin{subequations}
   \begin{align}
      \abs{\omega_1} &= \abs[\Big]{
         \frac{w_1}{z_1} - 1
      } = 0, \\
      \norm{\bm\upsilon_1} &= \norm[\Big]{
         \frac{\bfv_1}{z_1}
      } = \norm[\Big]{
         \frac{\bfy_1}{z_1}
      } \leq \delta(\bfA_1).
   \end{align}
\end{subequations}
Then~\eqref{eq: series bounds iterations} resolves into
\begin{subequations}
   \begin{align}
      \begin{bmatrix}
         a_N \\ b_N
      \end{bmatrix} &= \bfD_N \bfD_{N-1} \dots \bfD_2 \begin{bmatrix}
         1 \\ \delta(\bfA_1)
      \end{bmatrix} \\
      &= \bfD_N \bfD_{N-1} \dots \bfD_1 \begin{bmatrix}
         1 \\ 0
      \end{bmatrix}.
   \end{align}
\end{subequations}
Returning to $w$ and $z$,
\begin{equation}
   \abs[\Big]{
      \frac{w_N}{z_1 z_2 \dots z_N} - 1
   } \leq \underbrace{\begin{bmatrix}
      1 & 0
   \end{bmatrix} \bfD_N \bfD_{N-1} \dots \bfD_1 \begin{bmatrix}
      1 \\ 0
   \end{bmatrix} - 1}_d
\end{equation}
and, assuming $d < 1$,
\begin{equation}
   (1-d) \prod_{k=1}^N \abs{z_k}
   \leq \abs{w_N} \leq (1+d) \prod_{k=1}^N \abs{z_k}
\end{equation}
and
\begin{equation}
   \dist\p[\bigg]{
      \arg w_N, \: \sum_{k=1}^N \arg z_k
   } \leq \arcsin d
\end{equation}
which implies the bounds~\eqref{eq: bounds series}.

\section{Proof of~\eqref{eq: parallel bounds c inv} in Theorem~\ref{th: parallel}}
\label{ap: proof parallel}

Take arbitrary $\bfx \in \bbC^n \setminus \set{0}$ and let for $k = 1, 2, \dots, N$
\begin{subequations}
   \begin{align}
      z_k &= \frac{\bfx^* \bfA_k \bfx}{\norm{\bfx}^2} \in \calW(\bfA_k), \\
      w &= \frac{\bfy^* \bfB\inv \bfy}{\norm{\bfy}^2} \in \calW(\bfB\inv).
   \end{align}
\end{subequations}
We aim to lower-bound $\abs{w}$. With $\bfy = \bfB \bfx$
\begin{equation}
   \label{eq: proof parallel w}
   w = \frac{\bfx^* \bfB^* \bfx}{\norm{\bfB \bfx}^2}
   = \p*{\frac{\norm{\bfx}}{\norm{\bfB \bfx}}}^2 \p[\bigg]{
      \sum_{k=1}^N z_k
   }^*
\end{equation}
where
\begin{equation}
   \p*{\frac{\norm{\bfB \bfx}}{\norm{\bfx}}}^2
   = \sum_{k=1}^N \p*{\frac{\norm{\bfA_k \bfx}}{\norm{\bfx}}}^2
   + \sum_{k \neq \ell} \frac{\bfx^* \bfA_k^* \bfA_\ell^{} \bfx}{\norm{\bfx}^2}.
\end{equation}
Using the identities
\begin{gather}
   \bfx^* \bfA_k^* \bfA_\ell^{} \bfx
   = z_k^* z_\ell^{} \norm{\bfx}^2 + \bfx^* \bfA_k^* \bfP \bfA_\ell^{} \bfx, \\
   \sum_{k \neq \ell} z_k^* z_\ell^{}
   = \abs[\bigg]{\sum_{k=1}^N z_k}^2 - \sum_{k=1}^N \abs{z_k}^2, \\
   \p*{\frac{\norm{\bfA_k \bfx}}{\norm{\bfx}}}^2 - \abs{z_k}^2
   = \p*{\frac{\norm{\bfP \bfA_k \bfx}}{\norm{\bfx}}}^2
\end{gather}
where $\bfP$ is the projection matrix from Lemma~\ref{le: orthogonal bound}, we obtain
\begin{multline}
   \p*{\frac{\norm{\bfB \bfx}}{\norm{\bfx}}}^2
   = \abs[\bigg]{\sum_{k=1}^N z_k}^2
   + \sum_{k=1}^N \p*{\frac{\norm{\bfP \bfA_k \bfx}}{\norm{\bfx}}}^2 \\
   + \sum_{k \neq \ell} \p*{
      \frac{\bfP \bfA_k \bfx}{\norm{\bfx}}
   }^* \p*{
      \frac{\bfP \bfA_\ell \bfx}{\norm{\bfx}}
   }
\end{multline}
and by Lemma~\ref{le: orthogonal bound}
\begin{equation}
   \p*{\frac{\norm{\bfB \bfx}}{\norm{\bfx}}}^2
   \leq \abs[\bigg]{\sum_{k=1}^N z_k}^2
   + \p[\bigg]{\sum_{k=1}^N \delta(\bfA_k) \abs{z_k}}^2.
\end{equation}
Then from~\eqref{eq: proof parallel w}
\begin{subequations}
   \begin{align}
      \abs{w} &\geq \frac{1}{
         \abs[\Big]{\sum\limits_{k=1}^N z_k}
         + \dfrac{\p[\Big]{
            \sum\limits_{k=1}^N \delta(\bfA_k) \abs{z_k}
         }^2}{\abs[\Big]{\sum\limits_{k=1}^N z_k}}
      } \\
      &\geq \frac{1}{
         r(\bfB) + \dfrac{1}{c(\bfB)}\p[\Big]{
            \sum\limits_{k=1}^N \delta(\bfA_k) r(\bfA_k)
         }^2
      }
   \end{align}
\end{subequations}
implying~\eqref{eq: parallel bounds c inv}.

\section{Proof of Proposition~\ref{pr: normalized defect}}
\label{ap: proof normalized defect}

It is sufficient to show that Lemma~\ref{le: orthogonal bound} holds with $\delta$ replaced by $\delta_{\rm n}$. This is justified by replacing the estimation~\eqref{eq: problematic ratio} with
\begin{equation}
   \label{eq: improved ratio}
   \p*{\frac{\norm{\bfA\bfx}}{\norm{\bfx}}}^2
   = \frac{\p*{\dfrac{\abs{\bfx^* \bfA \bfx}}{\norm{\bfx}^2}}^2}
   {\p*{\dfrac{\abs{\bfx^* \bfA \bfx}}{\norm{\bfx}\norm{\bfA\bfx}}}^2}
   \leq \frac{\abs{z}^2}{c_{\rm n}(\bfA)^2}.
\end{equation}

The inequality $\delta_{\rm n}(\bfA) \leq \delta(\bfA)$ follows from
\begin{equation}
   c_{\rm n}(\bfA)^2
   = \p*{\frac{\abs{\bm\xi^* \bfA \bm\xi}}
   {\norm{\bm\xi} \norm{\bfA\bm\xi}}}^2
   = \frac{\abs{\bm\xi^* \bfA \bm\xi}}
   {\norm{\bm\xi}^2}
   \frac{\abs{\bm\xi^* \bfA \bm\xi}}
   {\norm{\bfA\bm\xi}^2}
   \geq c(\bfA) c(\bfA\inv)
\end{equation}
which holds with some $\bm\xi \in \bbC^n$.

If $\bfA = \gamma \bfU$ with $\gamma > 0$ and unitary $\bfU$ then $c(\bfA) = \gamma c(\bfU)$, $c(\bfA\inv) = \gamma\inv c(\bfU)$, $c_{\rm n}(\bfA) = c(\bfU)$, $c_{\rm n}(\bfA)^2 = c(\bfA) c(\bfA\inv)$, and $\delta_{\rm n}(\bfA) = \delta(\bfA)$.

\section{Proof of Lemma~\ref{le: segmental radius}}
\label{ap: proof segmental radius}

As noted in the main text, $\calW_{\rm n}(\cdot)$ is ``convex as seen from the origin\rlap{.}'' Thus, \eqref{eq: normalized Crawford} can be reformulated as
\begin{equation}
   c_{\rm n}(\bfA) = \min_{z \in \conv\calW_{\rm n}(\bfA)} \abs{z}
\end{equation}
or, with the dual representation $\abs{z} = \max\limits_{\theta\in[0,2\pi]} \real(\ee^{j\theta} z)$, as
\begin{equation}
   c_{\rm n}(\bfA) = \min_{z \in \conv\calW_{\rm n}(\bfA)} \:
   \max_{\theta\in[0,2\pi]} \real(\ee^{j\theta} z).
\end{equation}
By a minimax theorem (e.g., \cite[Theorem~3.1.30]{nesterovLecturesConvexOptimization2018}), the minimum and maximum operators can be swapped:
\begin{equation}
   \label{eq: normalized Crawford dual}
   c_{\rm n}(\bfA) = \max_{\theta\in[0,2\pi]} \:
   \min_{z \in \conv\calW_{\rm n}(\bfA)} \real(\ee^{j\theta} z).
\end{equation}
On the other hand, \eqref{eq: minimal segment} is equivalent to
\begin{equation}
   \tfrac12\len\Psi(\bfA) = \arccos \max_{\theta\in[0,2\pi]} \:
   \min_{z \in \conv\calW_{\rm n}(\bfA)} \real(\ee^{j\theta} z).
\end{equation}
Comparing the latter to~\eqref{eq: normalized Crawford dual}, we conclude~\eqref{eq: segmental radius}.

\bibliography{paper}

@book{zhouRobustOptimalControl1996,
  title = {Robust and optimal control},
  author = {Zhou, Kemin and Doyle, John C. and Glover, Keith},
  year = 1996,
  publisher = {Prentice Hall},
  address = {Upper Saddle River, NJ},
  isbn = {978-0-13-456567-5},
  langid = {english}
}

@article{zamesInputoutputStabilityTimevarying1966a,
  title = {On the input-output stability of time-varying nonlinear feedback systems -- {{Part I}}: {{Conditions}} derived using concepts of loop gain, conicity, and positivity},
  author = {Zames, G.},
  year = 1966,
  journal = {IEEE Transactions on Automatic Control},
  volume = {11},
  number = {2},
  pages = {228--238},
  note = {doi:10.1109/TAC.1966.1098316}
}

@article{huangGainPhaseDecentralized2024,
  title = {Gain and phase: decentralized stability conditions for power electronics-dominated power systems},
  shorttitle = {Gain and phase},
  author = {Huang, Linbin and Wang, Dan and Wang, Xiongfei and Xin, Huanhai and Ju, Ping and Johansson, Karl H. and D{\"o}rfler, Florian},
  year = 2024,
  journal = {IEEE Transactions on Power Systems},
  volume = {39},
  number = {6},
  pages = {7240--7256},
  note = {doi:10.1109/TPWRS.2024.3380528}
}

@book{vidyasagarInputoutputAnalysisLargescale1981,
  title = {Input-output analysis of large-scale interconnected systems},
  author = {Vidyasagar, M.},
  year = 1981,
  publisher = {Springer-Verlag},
  address = {Berlin/Heidelberg},
  note = {doi:10.1007/BFb0044060},
  isbn = {978-3-540-10501-5},
  langid = {english}
}

@book{vidyasagarNonlinearSystemsAnalysis1993,
  title = {Nonlinear systems analysis},
  author = {Vidyasagar, M.},
  year = 1993,
  publisher = {Prentice Hall},
  address = {Englewood Cliffs, N.J},
  isbn = {978-0-13-623463-0},
  lccn = {QA402 .V53 1993}
}

@article{chaffeyHomotopyTheoremIncremental2026,
  title = {A homotopy theorem for incremental stability},
  author = {Chaffey, Thomas and Kharitenko, Andrey and Forni, Fulvio and Sepulchre, Rodolphe},
  year = 2026,
  journal = {IEEE Transactions on Automatic Control},
  volume = {71},
  number = {4},
  pages = {2740--2745},
  note = {doi:10.1109/TAC.2025.3632433}
}

@misc{krebbekxGraphicalAnalysisNonlinear2025,
  title = {Graphical analysis of nonlinear multivariable feedback systems},
  author = {Krebbekx, Julius P. J. and T{\'o}th, Roland and Das, Amritam},
  year = 2025,
  number = {arXiv:2507.16513},
  eprint = {2507.16513},
  primaryclass = {eess},
  publisher = {arXiv},
  note = {doi:10.48550/arXiv.2507.16513},
  archiveprefix = {arXiv}
}

@article{wangPhasesComplexMatrix2020,
  title = {On the phases of a complex matrix},
  author = {Wang, Dan and Chen, Wei and Khong, Sei Zhen and Qiu, Li},
  year = 2020,
  journal = {Linear Algebra and its Applications},
  volume = {593},
  pages = {152--179},
  note = {doi:10.1016/j.laa.2020.01.035}
}

@misc{zhaoWhenSmallGain2022,
  title = {When small gain meets small phase},
  author = {Zhao, Di and Chen, Wei and Qiu, Li},
  year = 2022,
  number = {arXiv:2201.06041},
  eprint = {2201.06041},
  primaryclass = {eess},
  publisher = {arXiv},
  note = {doi:10.48550/arXiv.2201.06041},
  archiveprefix = {arXiv}
}

@book{hornTopicsMatrixAnalysis1991,
  title = {Topics in matrix analysis},
  author = {Horn, Roger A. and Johnson, Charles R.},
  year = 1991,
  publisher = {Cambridge University Press},
  note = {doi:10.1017/CBO9780511840371},
  isbn = {978-0-521-30587-7 978-0-521-46713-1 978-0-511-84037-1}
}

@article{stewartPertubationBoundsDefinite1979,
  title = {Perturbation bounds for the definite generalized eigenvalue problem},
  author = {Stewart, G. W.},
  year = 1979,
  journal = {Linear Algebra and its Applications},
  volume = {23},
  pages = {69--85},
  note = {doi:10.1016/0024-3795(79)90094-6}
}

@article{auzingerSectorialOperatorsNormalized2003,
  title = {Sectorial operators and normalized numerical range},
  author = {Auzinger, W.},
  year = 2003,
  journal = {Applied Numerical Mathematics},
  volume = {45},
  number = {4},
  pages = {367--388},
  note = {doi:10.1016/S0168-9274(02)00254-4}
}

@book{katoPerturbationTheoryLinear1984,
  title = {Perturbation theory for linear operators},
  author = {Kato, Tosio},
  year = 1984,
  edition = {2. corr. print. of the 2. ed},
  number = {132},
  publisher = {Springer},
  address = {Berlin Heidelberg},
  isbn = {978-0-387-07558-7 978-3-540-07558-5},
  langid = {english}
}

@article{druryNumericalRangeMatrix2024,
  title = {The numerical range of matrix products},
  author = {Drury, Stephen},
  year = 2024,
  journal = {The Electronic Journal of Linear Algebra},
  volume = {40},
  pages = {307--321},
  note = {doi:10.13001/ela.2024.8491},
  langid = {english}
}

@article{johnsonNumericalDeterminationField1978,
  title = {Numerical determination of the field of values of a general complex matrix},
  author = {Johnson, Charles R.},
  year = 1978,
  journal = {SIAM Journal on Numerical Analysis},
  volume = {15},
  number = {3},
  pages = {595--602},
  publisher = {{Society for Industrial and Applied Mathematics}},
  note = {doi:10.1137/0715039}
}

@article{chenCyclicSmallPhase2026,
  title = {A cyclic small phase theorem},
  author = {Chen, Chao and Chen, Wei and Zhao, Di and Chen, Jianqi and Qiu, Li},
  year = 2026,
  journal = {IEEE Transactions on Automatic Control},
  volume = {71},
  number = {3},
  pages = {1676--1691},
  note = {doi:10.1109/TAC.2025.3617287}
}

@article{hochstenbachNumericalApproximationField2013,
  title = {Numerical approximation of the field of values of the inverse of a large matrix},
  author = {Hochstenbach, Michiel E and Singer, David A and Zachlin, Paul F},
  year = 2013,
  journal = {CASA-report},
  volume = {1308},
  publisher = {Technische Universiteit Eindhoven},
  langid = {english}
}

@article{druryNumericalRadiusInequality2024,
  title = {A numerical radius inequality for sector operators},
  author = {Drury, Stephen},
  year = 2024,
  journal = {Linear Algebra and its Applications},
  volume = {687},
  pages = {108--116},
  note = {doi:10.1016/j.laa.2024.01.019}
}

@article{linsNormalizedNumericalRange2018,
  title = {The normalized numerical range and the {{Davis}}--{{Wielandt}} shell},
  author = {Lins, Brian and Spitkovsky, Ilya M. and Zhong, Siyu},
  year = 2018,
  journal = {Linear Algebra and its Applications},
  volume = {546},
  pages = {187--209},
  note = {doi:10.1016/j.laa.2018.01.027}
}

@article{tsypkinHighgainRobustControl1999,
  title = {High-gain robust control},
  author = {Tsypkin, {\relax Ya}. Z. and Polyak, B. T.},
  year = 1999,
  journal = {European Journal of Control},
  volume = {5},
  number = {1},
  pages = {3--9},
  note = {doi:10.1016/S0947-3580(99)70132-9}
}

@article{donoghueNumericalRangeBounded1957,
  title = {On the numerical range of a bounded operator},
  author = {Donoghue, William F.},
  year = 1957,
  journal = {Michigan Mathematical Journal},
  volume = {4},
  number = {3},
  note = {doi:10.1307/mmj/1028997958}
}

@article{bodeRelationsAttenuationPhase1940,
  title = {Relations between attenuation and phase in feedback amplifier design},
  author = {Bode, H. W.},
  year = 1940,
  journal = {The Bell System Technical Journal},
  volume = {19},
  number = {3},
  pages = {421--454},
  note = {doi:10.1002/j.1538-7305.1940.tb00839.x}
}

@book{hallAnalysisSynthesisLinear1943,
  title = {The analysis and synthesis of linear servomechanisms},
  author = {Hall, Albert C.},
  year = 1943,
  publisher = {The MIT Press},
  note = {doi:10.7551/mitpress/1260.001.0001},
  isbn = {978-0-262-31078-9},
  langid = {english}
}

@article{janssenMakingGraphsReducible1997,
  title = {Making graphs reducible with controlled node splitting},
  author = {Janssen, Johan and Corporaal, Henk},
  year = 1997,
  journal = {ACM Trans. Program. Lang. Syst.},
  volume = {19},
  number = {6},
  pages = {1031--1052},
  note = {doi:10.1145/267959.269971}
}

@book{nesterovLecturesConvexOptimization2018,
  title = {Lectures on convex optimization},
  author = {Nesterov, Yurii},
  year = 2018,
  volume = {137},
  publisher = {Springer International Publishing},
  address = {Cham},
  note = {doi:10.1007/978-3-319-91578-4},
  isbn = {978-3-319-91577-7 978-3-319-91578-4}
}

@book{ilchmannNonidentifierbasedHighgainAdaptive1993,
  title = {Non-identifier-based high-gain adaptive control},
  author = {Ilchmann, Achim},
  year = 1993,
  publisher = {Springer-Verlag},
  address = {London},
  note = {doi:10.1007/BFb0032266},
  isbn = {978-3-540-19845-1},
  langid = {english}
}

@article{kreinAngularLocalizationSpectrum1969,
  title = {Angular localization of the spectrum of a multiplicative integral in a {{Hilbert}} space},
  author = {Krein, M. G.},
  year = 1969,
  journal = {Functional Analysis and Its Applications},
  volume = {3},
  number = {1},
  pages = {73--74},
  note = {doi:10.1007/BF01078278},
  langid = {english}
}

@article{mauryaMarineVehiclePath2009,
  title = {Marine vehicle path following using inner-outer loop control},
  author = {Maurya, P. and Aguiar, A. Pedro and Pascoal, A.},
  year = 2009,
  journal = {IFAC Proceedings Volumes},
  series = {8th {{IFAC Conference}} on {{Manoeuvring}} and {{Control}} of {{Marine Craft}}},
  volume = {42},
  number = {18},
  pages = {38--43},
  note = {doi:10.3182/20090916-3-BR-3001.0071}
}

@article{bergerFunnelControlNonlinear2021,
  title = {Funnel control of nonlinear systems},
  author = {Berger, Thomas and Ilchmann, Achim and Ryan, Eugene P.},
  year = 2021,
  journal = {Mathematics of Control, Signals, and Systems},
  volume = {33},
  number = {1},
  pages = {151--194},
  note = {doi:10.1007/s00498-021-00277-z},
  langid = {english}
}

@article{dashkovskiySmallGainTheorems2010,
  title = {Small gain theorems for large scale systems and construction of {{ISS Lyapunov}} functions},
  author = {Dashkovskiy, Sergey N. and R{\"u}ffer, Bj{\"o}rn S. and Wirth, Fabian R.},
  year = 2010,
  journal = {SIAM Journal on Control and Optimization},
  volume = {48},
  number = {6},
  pages = {4089--4118},
  publisher = {{Society for Industrial and Applied Mathematics}},
  note = {doi:10.1137/090746483}
}

\end{document}